\documentclass[12pt]{article}
\usepackage{amsmath}
\usepackage{amsfonts}
\usepackage{booktabs}
\usepackage{mathtools}
\usepackage{amsthm}
\usepackage{bbm}
\usepackage{graphicx}
\usepackage{algorithm}
\usepackage{algpseudocode}
\usepackage{enumerate}
\usepackage{xcolor}
\usepackage{natbib}
\usepackage{url} % not crucial - just used below for the URL 

\newtheorem{lemma}{Lemma}

\newtheorem{proposition}{Proposition}
\newtheorem{theorem}{Theorem}
\newtheorem{definition}{Definition}
\newtheorem{example}{Example}
\newtheorem{assumption}{Assumption}

\newtheorem{remark}{Remark}

\DeclareMathOperator*\argmin{arg\,min}
\DeclareMathOperator*\argmax{arg\,max}

\newcommand{\mE}{\mathbb{E}}

\newcommand{\cL}{\mathcal{L}}

\newcommand{\mP}{\mathbb{P}}

\newcommand{\cQ}{\mathcal{Q}}

\renewcommand{\thefootnote}{\fnsymbol{footnote}}
\def\spacingset#1{\renewcommand{\baselinestretch}%
{#1}\small\normalsize} \spacingset{1}

\title{\bf Fair Regression under Demographic Parity: A Unified Framework}
\author{
  Yongzhen Feng$^{\dagger}$\\
  Center for Statistical Science, Tsinghua University\\
  Weiwei Wang$^{\dagger}$\\
  Department of Statistics, Texas A\&M University \\
  Raymond K.~W.~Wong$^{\dagger}$ \\
  Department of Statistics, Texas A\&M University \\
  and\\
  Xianyang Zhang$^{\dagger}$\thanks{Corresponding author: zhangxiany@stat.tamu.edu.}\\
  Department of Statistics, Texas A\&M University
}
\date{}

\begin{document}
\maketitle

\renewcommand{\thefootnote}{\fnsymbol{footnote}}
\footnotetext[2]{Authors are listed in alphabetical order.}

\bigskip
\begin{abstract}
We propose a unified framework for fair regression tasks formulated as risk minimization problems subject to a demographic parity constraint. Unlike many existing approaches that are limited to specific loss functions or rely on challenging non-convex optimization, our framework is applicable to a broad spectrum of regression tasks. Examples include linear regression with squared loss, binary classification with cross-entropy loss, quantile regression with pinball loss, and robust regression with Huber loss. We derive a novel characterization of the fair risk minimizer, which yields a computationally efficient estimation procedure for general loss functions. Theoretically, we establish the asymptotic consistency of the proposed estimator and derive its convergence rates under mild assumptions. We illustrate the method's versatility through detailed discussions of several common loss functions. Numerical results demonstrate that our approach effectively minimizes risk while satisfying fairness constraints across various regression settings.
\end{abstract}

\vfill

\newpage
\spacingset{1.9} % DON'T change the spacing!
\section{Introduction} \label{sec:intro}
Algorithmic fairness is a rapidly evolving domain aimed at identifying and rectifying systematic disadvantages faced by specific demographic groups—defined by sensitive attributes such as race or gender—within machine learning systems. For comprehensive reviews, refer to \cite{mehrabi2021survey,pessach2022review,caton2024fairness}. In high-stakes decision-making scenarios, such as parole determination, it is imperative that algorithms do not propagate bias by relying on sensitive information. Consequently, significant research effort is dedicated to developing methods that satisfy fairness constraints without compromising prediction performance.

While numerous notions of algorithmic fairness exist (see Section 4.1 of \cite{mehrabi2021survey} for a summary), we focus on demographic parity, a widely adopted criterion requiring the probability distribution of the predicted output to be independent of sensitive attributes. Prior work has predominantly focused on developing fair algorithms for specific loss functions in classification and regression contexts \citep{agarwal2018reductions,agarwal2019fair,chzhen2020fair,silvia2020general,liu2022conformalized}. These approaches, however, typically necessitate a case-by-case analysis. In this work, we propose a unified framework for fair regression formulated as risk minimization under a demographic parity constraint. Our approach generalizes existing findings and applies to a broad spectrum of regression tasks, including linear regression with squared loss, binary classification with cross-entropy loss, quantile regression with pinball loss, and robust regression with Huber loss. Central to our method is a novel characterization of the optimal fair prediction function rooted in optimal transport theory. This theoretical insight yields a simple and computationally efficient estimation algorithm that accommodates this wide variety of loss functions.

Our contributions are two-fold. First, utilizing optimal transport theory, we demonstrate that the optimal fair prediction functions for various regression problems belong to a specific function class (defined in \eqref{eq-opt-form}). This finding, based on a connection between fair regression and the Kantorovich barycenter problem, offers a new interpretation of the optimal fair predictor and allows us to address tasks previously unexplored in the literature, such as fair quantile, robust, and Poisson regression. Second, we develop a two-step estimation procedure. The first step involves estimating the optimal predictor for each sensitive group independently, enabling the use of standard off-the-shelf algorithms (e.g., kernel ridge regression, random forests, neural networks). In the second step, we synthesize these individual predictors into a fair predictor via a common estimated quantile function. We establish the asymptotic consistency of our estimator and derive its specific convergence rates. Extensive numerical studies demonstrate the proposed method's effectiveness in minimizing risk while maintaining fairness across diverse regression settings.

\subsection{Related works}
Recent years have seen a proliferation of methods for fair regression under a variety of fairness notions. Most closely related to our work is \citet{chzhen2020fair}, who propose a post-processing approach and derive a closed-form expression for the optimal fair predictor under squared loss. This characterization leads to a plug-in estimator that avoids directly solving a non-convex fairness-constrained optimization problem. However, closed-form solutions of this type are largely limited to a small class of losses—most notably squared loss and, in some settings, cross-entropy loss \citep{gaucher2023fair}; see Section \ref{sec:fair_para} for further discussion. \citet{li2024multi} extend \citet{chzhen2020fair} to multi-dimensional predictors, but still focus on squared loss. While squared loss is a natural choice when the target is the conditional mean, it can be inadequate when robustness, tail behavior, or asymmetric error is central, motivating fair regression methods that accommodate other loss functions. More broadly, much of the recent fair regression literature continues to center on squared loss \citep{fallah2025statistical,fukuchi2025meta,hu2023fairness,taturyan2024regression,fukuchi2023demographic}, and fair regression under alternative risks remains comparatively underexplored.

Beyond these standard losses, \citet{liu2022conformalized} introduce a method for fair quantile regression. Instead of defining the fair predictor as a minimizer of the pinball-loss risk under a fairness constraint, they define it as the closest function (in $L_2$ distance) to an unconstrained quantile estimator, which enables an estimation strategy reminiscent of \citet{chzhen2020fair}. As a consequence, their estimator does not generally minimize the fairness-constrained pinball-loss risk. As we show in Example \ref{ex-lad}, the fair risk minimizer under pinball loss typically does not admit a closed-form expression, which precludes a direct plug-in approach.
In the context of generalized linear models, \citet{do2022fair} develop an in-process algorithm, but their fairness targets—Equalized Expected Outcome and Equalized Expected Log-likelihood—differ from the demographic parity notion considered here.

Finally, \citet{agarwal2019fair} propose a reduction-based approach for enforcing demographic parity with Lipschitz losses. A key limitation is the need to discretize continuous responses, which can introduce approximation error and additional tuning. In contrast, our framework directly accommodates continuous responses without discretization. Relatedly, several works formulate fair regression as numerical optimization problems \citep{carrizosa2025fair,deza2026fair}; these approaches can require substantial hyperparameter tuning and often yield solutions that are less transparent than methods based on explicit characterizations.

\section{Problem setup}\label{sec:problem_setup}
Let $(X,S)\in\mathcal{X} \times \mathcal{S}$ denote the predictive variables and $Y\in \mathbb{R}$ the response variable.
Here, $S$ represents a protected attribute (e.g., marital status, age, or socioeconomic status), and $X$ comprises the remaining predictive features. We assume $\mathcal{X}\subseteq\mathbb{R}^d$.
For simplicity, we focus on settings with a binary protected variable, assuming $\mathcal{S}=\{0,1\}$ without loss of generality.
Let $P(S=s)=r_s$ for $s=0,1$.
The proposed approach extends straightforwardly to cases with multiple protected classes.
In this paper, we focus on the fairness notion of demographic parity.

\begin{definition}[Demographic parity]
A prediction function $\eta: \mathbb{R}^d \times \mathcal{S} \rightarrow \mathbb{R}$ satisfies demographic parity if
$$
\sup _{t \in \mathbb{R}}\left|\mathbb{P}(\eta(X, S) \leq t \mid S=0)-\mathbb{P}\left(\eta(X, S) \leq t \mid S=1\right)\right|=0.
$$
\end{definition}

For clarity, we introduce the conditional random variables $X^0$ and $X^1$ such that $X \overset{d}{=} X^0$ when $S=0$ and $X \overset{d}{=} X^1$ when $S=1$, where ``$\overset{d}{=}$'' denotes equality in distribution.

Our objective is to learn a prediction function $f:\mathcal{X}\times \mathcal{S}\rightarrow \mathbb{R}$ that minimizes prediction error while satisfying demographic parity, i.e., $f(X^0,0)\overset{d}{=} f(X^1,1)$.
Formally, we seek the fair risk minimizer defined by:
\begin{equation} \label{FP}
\argmin_{f: f(X^0,0)\overset{d}{=} f(X^1,1)}\mathbb{E}[\mathcal{L}(f(X,S),Y)],
\end{equation}
where $\mathcal{L}:\mathbb{R}\times\mathbb{R}\rightarrow\mathbb{R}_{\ge 0}$ is a loss function. As detailed later, our framework accommodates a wide range of loss functions.
We refer to the optimization problem \eqref{FP} as the Fair Regression Problem (FRP) and its solution as the optimal fair prediction function.

It is worth noting that the empirical realization of the fairness constraint typically introduces non-convexity, posing significant computational challenges that often necessitate approximations or specialized algorithmic solutions \citep{Donini-Oneto-Ben-David18, komiyama2018nonconvex, wu2019convexity}.

\begin{remark}[Extension to Multiple Groups]
{\rm 
    While we restrict our exposition to the binary case ($K=2$) for notational clarity, the proposed framework generalizes naturally to situations with multiple protected groups $\mathcal{S}=\{0, 1, \dots, K-1\}$. The definition of demographic parity extends to requiring that $f(X^s, s)$ have the same distribution for all $s \in \mathcal{S}$. As discussed in Section \ref{sec:erm} and Appendix \ref{app:multi_group}, our theoretical characterization and estimation procedure remain valid for general $K$.
}
\end{remark}

\section{Characterization of optimal fair prediction functions} \label{sec:fair_para}

To address the computational challenges posed by fairness constraints and to accommodate general loss functions, we propose a novel characterization of the fair risk minimizer. We begin by introducing a fundamental assumption maintained throughout the paper.

\begin{assumption}\label{assum:f_star}
There exists a function $f^*:\mathcal{X}\times \mathcal{S}\rightarrow \mathbb{R}$ such that the response variable $Y$ depends on the predictive variables $(X,S)$ only through $f^*(X,S)$, i.e., $Y \perp (X,S) \mid f^*(X,S)$.
\end{assumption}

Assumption \ref{assum:f_star} is mild and is satisfied in many common regression settings. For instance, it holds for additive noise models $Y=f^*(X,S)+\varepsilon$, where $\varepsilon$ is independent of $(X,S)$ \citep{chzhen2020fair, agarwal2019fair, liu2022conformalized}. More generally, it encompasses data-generating processes of the form $Y|(X,S)\sim F(\cdot|\theta)$ with $\theta=f^*(X,S)$. Examples include binary regression, where $F(\cdot|\theta)$ is a Bernoulli distribution with success probability $\theta$, and Poisson regression, where $F(\cdot|\theta)$ is a Poisson distribution with rate $\theta$ \citep[e.g.,][]{do2022fair}.

Under Assumption \ref{assum:f_star}, the conditional risk can be reformulated as:
\begin{align}\label{eq-cost}
	\mathbb{E}[\mathcal{L}(f(X,S),Y)|X,S] &= \int \mathcal{L}(f(X,S),y) dF(y|f^*(X,S)) \nonumber \\
	&:= \mathcal{C}(f^*(X,S),f(X,S)),
\end{align}
where $F(\cdot|f^*(X,S))$ denotes the conditional distribution of $Y$ given $(X,S)$, which depends solely on $f^*(X,S)$.

Let $\nu_{f}$ and $\nu_{f^*}$ denote the distributions of $f(X,S)$ and $f^*(X,S)$, respectively. Furthermore, let $\nu_{f^*|s}$ be the distribution of $f^*(X^s,s)$ for $s=0,1.$ We demonstrate that the FRP \eqref{FP} shares a fundamental connection with the Kantorovich problem in optimal transport theory, using the cost function $\mathcal{C}$ defined in \eqref{eq-cost}. Specifically, the Kantorovich problem associated with cost $\mathcal{C}$ is defined as:
\begin{align*}
	\inf_{\gamma\sim \Pi(\nu_{f^*},\nu) }\int\mathcal{C}(x,y)d\gamma(x,y)=\inf_{\gamma\sim \Pi(\nu_{f^*},\nu) }\mathbb{E}_{(U^*,U)\sim \gamma}[\mathcal{C}(U^*,U)],
\end{align*}
where the infimum is taken over $\Pi(\nu_{f^*},\nu)$, the set of all joint distributions on $\mathbb{R}\times \mathbb{R}$ with marginals $\nu_{f^*}$ and $\nu$.

\begin{assumption}\label{assum:cost_function}
The distribution $\nu_{f^*}$ is non-atomic (e.g., absolutely continuous with respect to the Lebesgue measure), and the cost function $\mathcal{C}(\cdot,\cdot)$ is continuous.
\end{assumption}

Under Assumption \ref{assum:cost_function}, Theorem B of \cite{pratelli2007equality} guarantees that the solution to the Kantorovich problem is given by $\gamma=(I,T)_{\#}\nu_{f^*},$ where $T$ is the optimal transport map and $T_{\#}\nu_{f^*}=\nu$. Here, $T_{\#}\nu_{f^*}$ denotes the pushforward measure defined by $T_{\#}\nu_{f^*}(E)=\nu_{f^*}(T^{-1}(E))$ for any measurable set $E$. Consequently, we have:
\begin{align*}
	\inf_{\gamma\sim \Pi(\nu_{f^*},\nu) }\int\mathcal{C}(x,y)d\gamma(x,y)=\mathbb{E}[\mathcal{C}(U^*,T(U^*))],
\end{align*}
where $U^*\sim \nu_{f^*}.$

Proposition \ref{prop1} below characterizes the solution to the FRP by establishing a link between the FRP and the corresponding Kantorovich Barycenter Problem (KBP), defined as:
\begin{align}\label{eq-kbp}
	\inf_{\nu}\sum_{s=0}^1 r_s\inf_{\gamma\sim \Pi(\nu_{f^*|s},\nu) }\mathbb{E}_{(U^*,U)\sim \gamma}[\mathcal{C}(U^*,U)],
\end{align}
where the infimum is over all probability distributions on $\mathbb{R}$. The KBP generalizes the Wasserstein barycenter problem \citep{agueh2011barycenters}. Let $\bar{\nu}$ be the solution to the KBP in \eqref{eq-kbp}. Consider the Kantorovich problem $\inf_{\gamma\sim \Pi(\nu_{f^*|s},\bar{\nu}) }\mathbb{E}_{(U^*,U)\sim \gamma}[\mathcal{C}(U^*,U)]$ and let $T_s$ be the corresponding optimal transport map such that:
$$
\inf_{\gamma\sim \Pi(\nu_{f^*|s},\bar{\nu}) }\mathbb{E}_{(U^*,U)\sim \gamma}[\mathcal{C}(U^*,U)] =\mathbb{E}[\mathcal{C}(U^*,T_s(U^*))]
$$
for $U^*\sim \nu_{f^*|s}.$ We remark that the existence of the optimal transport maps $T_s$ is guaranteed because the non-atomicity of the marginal distribution $\nu_{f^*}$ (Assumption \ref{assum:cost_function}) implies that the conditional distributions $\nu_{f^*|s}$ are also non-atomic for all $s$ with $r_s > 0$.

\begin{proposition}\label{prop1}
Under Assumptions \ref{assum:f_star} and \ref{assum:cost_function}, the solution to the FRP belongs to the class:
	\begin{align*}
	\mathcal{G}=&\Big\{T_{s}\circ f^{\ast}(x,s): \text{$T_0$ and $T_1$ are transport maps satisfying }
    T_0\circ f^*(X^0,0) \overset{d}{=} T_1\circ f^*(X^1,1)\Big\}.
	\end{align*}
\end{proposition}

Proposition \ref{prop1} establishes that the optimal fair prediction function is governed by the transport maps inherent to the KBP. To operationalize this result and facilitate the practical implementation detailed in Section \ref{sec:erm}, we seek a more explicit characterization of $T_0$ and $T_1$. The following assumption on the cost function $\mathcal{C}$ ensures that these transport maps admit a computationally tractable structure.

\begin{assumption}\label{assum:cost_function2}
    The cost function $\mathcal{C}(x,y)$ is twice-differentiable and satisfies the condition:
	 \begin{align}\label{con-twice}
		\frac{\partial^2 \mathcal{C}(x,y)}{\partial x\partial y}< 0,
	\end{align}
	for all $x, y \in \mathbb{R}.$ Furthermore, we assume that the distributions $\nu_{f^*|s}$ are non-atomic for $s=0,1$.
\end{assumption}

In Appendix \ref{sec:proof}, we demonstrate that under these conditions, the transport maps can be decomposed as $T_s = Q_{\bar{\nu}}\circ F^*_s$, where $F^*_s$ denotes the cumulative distribution function (CDF) of $f^*(X^s,s)$, and $Q_{\bar{\nu}}$ is the quantile function associated with the barycenter $\bar{\nu}$.
For the specific case where $\mathcal{C}(x,y)=h(x-y)$ for a strictly convex function $h$, a rigorous proof is provided in Theorem 2.9 of \cite{santambrogio2015optimal}. We extend this result to general cost functions satisfying Condition \eqref{con-twice} via similar arguments; see the proof of Proposition \ref{prop2} in Appendix \ref{sec:proof}.

\begin{proposition}\label{prop2}
Under Assumptions \ref{assum:f_star}, \ref{assum:cost_function} and \ref{assum:cost_function2}, the solution to the FRP belongs to the class
	\begin{align}\label{eq-opt-form}
		\widetilde{\mathcal{G}}=\left\{Q \circ F_{s}^*\circ f^*(x,s): \text{$Q$ is a quantile function}\right\}.
	\end{align}
\end{proposition}

A key advantage of the characterization in Proposition \ref{prop2} is that any predictor in the class $\widetilde{\mathcal{G}}$ naturally satisfies the demographic parity constraint. Consequently, the original constrained FRP reduces to an unconstrained risk minimization over the set of quantile functions. By dropping the explicit fairness constraint, we identify the optimal quantile function $\widetilde{Q}$ as the minimizer of the risk within this class:
$$
\widetilde{Q} := \argmin_{Q}\mathbb{E}[\mathcal{L}(Q \circ F_{S}^*\circ f^*(X,S),Y)],$$
where the minimization is over all quantile functions. In specific regression settings, such as fair linear regression with squared loss (see Example \ref{ex4}) or binary classification with cross-entropy loss (see Example \ref{ex-cross-solution}), $\widetilde{Q}$ admits a closed-form solution. \cite{chzhen2020fair} leveraged this property for the squared loss case, identifying $\widetilde{Q}$ as a Wasserstein barycenter and proposing a plug-in estimation approach. However, a closed-form solution is not always derivable (as illustrated in Example \ref{ex-lad}), which limits the applicability of such plug-in strategies to specific loss functions. Furthermore, even when closed forms exist, they typically depend on unknown population quantities, necessitating problem-specific estimation strategies and theoretical analyses.
In contrast, Proposition \ref{prop2} offers a general characterization applicable to a broad spectrum of loss functions. This result facilitates the development of the unified estimation framework presented in Section \ref{sec:erm}.

Before describing the estimation procedure, we verify that Assumption \ref{assum:cost_function2} is mild and holds in many commonly encountered regression settings, as illustrated by the following examples.

\begin{example}[Regression with convex loss]\label{ex1}
{\rm Suppose $Y=f^*(X,S)+\varepsilon$ where $\varepsilon\sim G$ is independent of $(X,S)$, and the loss function is given by $\mathcal{L}(f(X,S),Y)=\mathcal{L}(y-f(X,S))$ for some strictly convex function $\mathcal{L}$. Then,
\begin{align*}
\mathcal{C}(f^*(X,S),f(X,S))=\int \mathcal{L}(f^*(X,S)-f(X,S)+u)d G(u)=h(f^*(X,S)-f(X,S)),
\end{align*}
where $h(x)=\int \mathcal{L}(x+u)dG(u)$ is strictly convex due to the strict convexity of $\mathcal{L}$. Consequently, Assumption \ref{assum:cost_function2} is satisfied.
}
\end{example}

\begin{example}[Binary classification with cross-entropy loss]\label{ex-cross}
{\rm
Suppose that conditional on $(X,S)$, $Y$ follows a Bernoulli distribution with success probability $f^*(X,S)\in (0,1)$. Consider the cross-entropy loss
$\mathcal{L}(f(X,S),Y)=-Y\log f(X,S)-(1-Y)\log (1-f(X,S))$. This implies:
\begin{align*}
\mathcal{C}(f^*(X,S),f(X,S))=-f^*(X,S)\log f(X,S)-(1-f^*(X,S))\log (1-f(X,S)).
\end{align*}
It follows that $\frac{\partial^2 \mathcal{C}(x,y)}{\partial x\partial y}=-\{y(1-y)\}^{-1}<0$ for any $x,y\in (0,1)$, satisfying Assumption \ref{assum:cost_function2}.
}
\end{example}

\begin{example}[Generalized regression]\label{ex-expon}
{\rm
Suppose $Y|X,S \sim h(Y)\exp\left(T(Y)\theta -A(\theta)\right)$ belongs to an exponential family with canonical parameter $\theta$. Let $g(\theta)=f^*(X,S)$, where $g$ is a known, strictly monotonic link function with inverse $g^{-1}$. Consider the negative log-likelihood loss (up to an additive constant):
$\mathcal{L}(f(X,S),Y)=A\circ g^{-1}(f(X,S))-T(Y) g^{-1}(f(X,S))$.
Using the fact that $\mathbb{E}[T(Y)]=A^{(1)}(\theta)$ and $\theta = g^{-1}(f^*(X,S))$, we obtain:
\begin{align*}
\mathcal{C}(f^*(X,S),f(X,S))=&A\circ g^{-1}(f(X,S))-A^{(1)}(g^{-1}(f^*(X,S))) g^{-1}(f(X,S)).
\end{align*}
Differentiation yields:
\begin{align*}
\frac{\partial^2 \mathcal{C}(x,y)}{\partial x\partial y}=-A^{(2)}(g^{-1}(x))\frac{\partial g^{-1}(x)}{\partial x} \frac{\partial g^{-1}(y)}{\partial y}.
\end{align*}
Since $A^{(2)}(x)>0$ (variance function) and $g^{-1}$ is strictly monotonic, we have $\frac{\partial^2 \mathcal{C}(x,y)}{\partial x\partial y} <0$, satisfying Assumption \ref{assum:cost_function2}.
}
\end{example}

The next proposition shows that the population risk can be rewritten in terms of the induced cost $\mathcal{C}$, yielding a pointwise characterization of the population minimizer $\widetilde{Q}(u)$.

\begin{proposition}[Risk decomposition and pointwise characterization of $\widetilde{Q}$]\label{prop:risk-decomp-pointwise}
Recall the induced cost $\mathcal{C}$ in \eqref{eq-cost} and the random variable $U$ defined in Section~\ref{sec:fair_para}.
Then for any candidate $Q\in\mathcal{Q}$,
\begin{equation}\label{eq:risk-decomp}
\mathbb{E}\!\left[\mathcal{L}(Q(U),Y)\right]
=\mathbb{E}_{U}\!\left[\sum_{s=0}^{1} r_s\,\mathcal{C}\!\left(Q_s^*(U),\,Q(U)\right)\right].
\end{equation}
Consequently, any population minimizer $\widetilde{Q}$ satisfies, for Lebesgue-a.e.\ $u\in[0,1]$,
\begin{equation}\label{eq:Q-pointwise}
\widetilde{Q}(u)\in\arg\min_{q\in\mathbb{R}}\sum_{s=0}^{1} r_s\,\mathcal{C}\!\left(Q_s^*(u),\,q\right).
\end{equation}
\end{proposition}

\begin{remark}[Mixture view]\label{rm:mixture-view}
{\rm
Equivalently, if $\widetilde{Y}_u$ has conditional distribution
$\sum_{s=0}^{1} r_s\,F(\cdot\mid Q_s^*(u))$, then
$\sum_{s=0}^{1} r_s\,\mathcal{C}(Q_s^*(u),q)=\mathbb{E}[\mathcal{L}(q,\widetilde{Y}_u)]$.
}
\end{remark}

We now apply Proposition~\ref{prop:risk-decomp-pointwise} to the specific examples above.

\begin{example}[Regression with $l_2$ loss]\label{ex4}
{\rm
Specialize Example~\ref{ex1} to $\mathcal{L}(t)=t^2$ (squared loss). Then
$\mathcal{C}(x,q)=(q-x)^2+\sigma^2$, so by Proposition~\ref{prop:risk-decomp-pointwise},
for a.e.\ $u$,
\[
\widetilde{Q}(u)\in \arg\min_{q\in\mathbb{R}}\sum_{s=0}^1 r_s (q-Q_s^*(u))^2,
\]
whose unique minimizer is
\[
\widetilde{Q}(u)=r_0Q_0^*(u)+r_1Q_1^*(u).
\]
This recovers the closed form in Theorem~2.3 of \cite{chzhen2020fair}.
}
\end{example}

\begin{example}[Least absolute deviation regression]\label{ex-lad}
{\rm
In Example~\ref{ex1} with $\mathcal{L}(t)=|t|$ (LAD loss), Proposition~\ref{prop:risk-decomp-pointwise} yields, for a.e.\ $u$,
\[
\widetilde{Q}(u)\in \arg\min_{q\in\mathbb{R}}\sum_{s=0}^1 r_s\,\mathbb{E}\big|q-Q_s^*(u)-\varepsilon\big|.
\]
Equivalently, $\widetilde{Q}(u)$ is a median of the mixture distribution
$r_0\,G(\cdot+Q_0^*(u)) + r_1\,G(\cdot+Q_1^*(u))$, where $G$ is the CDF of $\varepsilon$.
In general this median does not admit a closed form.
}
\end{example}

\begin{example}[Binary classification with cross-entropy loss]\label{ex-cross-solution}
{\rm
In the setup of Example~\ref{ex-cross}, the induced cost is
$\mathcal{C}(x,q)=-x\log q-(1-x)\log(1-q)$.
Thus Proposition~\ref{prop:risk-decomp-pointwise} gives, for a.e.\ $u$,
\[
\widetilde{Q}(u)\in \arg\min_{q\in(0,1)}\sum_{s=0}^1 r_s\Big\{-Q_s^*(u)\log q-\big(1-Q_s^*(u)\big)\log(1-q)\Big\}.
\]
The first-order condition yields the unique solution
\[
\widetilde{Q}(u)=r_0Q_0^*(u)+r_1Q_1^*(u),
\]
consistent with \cite{gaucher2023fair}.
}
\end{example}

\begin{example}[Generalized regression]\label{ex-expon-solution}
{\rm
Continuing from Example~\ref{ex-expon}, the induced cost is
$\mathcal{C}(x,q)=A\!\circ g^{-1}(q)-A^{(1)}\!\big(g^{-1}(x)\big)\,g^{-1}(q)$.
Hence Proposition~\ref{prop:risk-decomp-pointwise} implies that for a.e.\ $u$, $\widetilde{Q}(u)$ solves
\[
\widetilde{Q}(u)\in \arg\min_{q\in\mathbb{R}}
\Big\{A\!\circ g^{-1}(q)-\Big(\sum_{s=0}^1 r_s A^{(1)}\!\big(g^{-1}(Q_s^*(u))\big)\Big) g^{-1}(q)\Big\}.
\]
Differentiating in $q$ and using that $g^{-1}$ is strictly monotone gives
\[
A^{(1)}\!\big(g^{-1}(\widetilde{Q}(u))\big)=\sum_{s=0}^1 r_s A^{(1)}\!\big(g^{-1}(Q_s^*(u))\big),
\]
so
\[
\widetilde{Q}(u)=g\circ (A^{(1)})^{-1}\!\left(\sum_{s=0}^1 r_s A^{(1)}\!\big(g^{-1}(Q_s^*(u))\big)\right).
\]
}
\end{example}

\begin{remark}\label{rm:known-unknown-cost}
{\rm 
Proposition~\ref{prop:risk-decomp-pointwise} characterizes $\widetilde{Q}$ pointwise via
$\widetilde{Q}(u)\in\arg\min_{q}\sum_{s=0}^{1} r_s\,\mathcal{C}\!\left(Q_s^*(u),q\right)$.
This leads to two regimes. \textbf{(i) Known induced cost.}
If the conditional law $Y\mid f^*(X,S)$ is specified enough that $\mathcal{C}(x,q)=\mathbb{E}[\mathcal{L}(q,Y)\mid f^*(X,S)=x]$
has a closed form (up to known constants), then one can estimate $Q_s^*$ and solve the above pointwise problem by plug-in; see e.g., \cite{chzhen2020fair}. \textbf{(ii) Unknown induced cost.}
In general, $x=f^*(X,S)$ may only identify a functional of $Y\mid X,S$ (rather than the full conditional distribution), so
$\mathcal{C}(x,q)$ is not identifiable as a function of $(x,q)$ alone. This is typical for robust/quantile settings, e.g.,
pinball (quantile) loss and Huber loss, where many conditional distributions share the same target $x$.
Our proposed estimator in the next section remains applicable in this regime because it does not require an explicit closed form for $\mathcal{C}$.
}
\end{remark}

\section{Empirical risk minimization with fairness constraints}\label{sec:erm}
Enlightened by the characterization in the previous section, we propose a two-step estimation procedure for finding the optimal fair prediction function.

\textbf{Step 1: Estimation of Group-wise Predictors.}
First, we estimate the conditional regression function $f^*(\cdot,s)$ using a consistent estimator $\widehat{f}(\cdot,s)$. To ensure theoretical independence, we assume the availability of an auxiliary training sample $\mathcal{D}^{est} = \{(X^{est}_j, S^{est}_j, Y^{est}_j)\}_{j=1}^{m}$, separate from the sample used in the subsequent step.
We perform standard regression separately for each protected group $s \in \{0, 1\}$ using $\mathcal{D}^{est}$.

Next, consider the primary training sample $\mathcal{D} = \{(X_i,S_i,Y_i): i=1,2,\dots,n\}$. Let $\widehat{F}_s$ be the empirical distribution function of the predicted values $\{\widehat{f}(X_i,S_i): S_i=s\}$ computed on this primary sample. For each observation $i \in \{1, \dots, n\}$ such that $S_i=s$, we compute the transformed variable:
$
\widehat{U}_i= \widehat{F}_{s}\circ \widehat{f}(X_i,S_i).
$

\textbf{Step 2: Estimation of the Common Quantile Function.}
In the second step, we estimate the optimal quantile function $\widetilde{Q}$ by solving the following empirical risk minimization problem:
\begin{align}\label{eq-est-Q}
\widehat{Q}=\argmin_{Q\in \widetilde {\mathcal{Q}}}\sum_{i=1}^n\mathcal{L}(Q(\widehat{U}_i),Y_i),
\end{align}
where $\widetilde {\mathcal{Q}}$ is a pre-specified class of functions defined on $[0,1]$ (e.g., non-decreasing functions).
Finally, the estimated optimal fair prediction function is constructed as:
$$\widehat{f}_{\text{fair}}(x,s):=\widehat{Q}\circ \widehat{F}_{s}\circ \widehat{f}(x,s).$$

\textbf{Extension to Multiple Groups.}
It is worth noting that this two-step procedure applies without modification to cases with $K > 2$ protected groups. In Step 1, one simply estimates $K$ distinct regression functions $\widehat{f}(\cdot, s)$. In Step 2, the optimization \eqref{eq-est-Q} remains a risk minimization over $\widetilde{\mathcal{Q}}$, where the empirical loss naturally aggregates data from all $K$ groups via the transformed samples $\{\widehat{U}_i\}_{i=1}^n$.

\subsection{Group-wise estimators} \label{sec:consistency_step1}
In our framework, the first step involves solving the regression problem on the auxiliary sample $\mathcal{D}^{est}$:
\begin{align} \label{first_f}
\widehat{f}(\cdot,s)=\argmin_{f\in\mathcal{F}_s}\sum_{j:S^{est}_j=s}\mathcal{L}(f(X^{est}_j,S^{est}_j),Y^{est}_j),
\end{align}
where $\mathcal{F}_s$ is a user-specified function class. This procedure yields a reasonable estimate for $f^*(\cdot,s)$ because $f^*(\cdot,s)$ is the minimizer of the population conditional risk for the loss functions discussed in Section \ref{sec:fair_para}. Formally, we observe that for each sensitive group $s$:
$$
f^*(\cdot,s) \in \argmin_{f\in\mathcal{F}_s}\mathbb{E}[\mathcal{L}(f(X,S),Y)|S=s].
$$
We provide detailed derivations for specific examples below.

\textbf{Binary Classification (Cross-Entropy Loss).}
Following Example \ref{ex-cross}, the population risk conditional on $S=s$ is:
\begin{align*}
    &\argmin_{f\in\mathcal{F}_s}\mathbb{E}[\mathcal{L}(f(X,S),Y)|S=s] \\
    &= \argmin_{f\in\mathcal{F}_s}\mathbb{E}_X\left[\mathbb{E}_Y[\mathcal{L}(f(X,S),Y)|X, S=s]\right] \\
    &= \argmin_{f\in\mathcal{F}_s}\mathbb{E}\left[-f^*(X,S)\log f(X,S) - (1-f^*(X,S))\log(1-f(X,S)) \Big| S=s\right].
\end{align*}
Conditioning further on $X=x$, the inner expression is minimized when $f(x,s) = f^*(x,s)$, which corresponds to the true conditional probability. Thus, if $f^*(\cdot, s)\in \mathcal{F}_s$, it is the minimizer of the objective.

\textbf{Poisson Regression.}
Consider Example \ref{ex-expon} with the canonical link (where $g$ is the identity and $A(\cdot) = \exp(\cdot)$). The optimization becomes:
\begin{align*}
    &\argmin_{f\in\mathcal{F}_s}\mathbb{E}[\mathcal{L}(f(X,S),Y)|S=s] \\
    &= \argmin_{f\in\mathcal{F}_s}\mathbb{E}\left[\exp(f(X, S)) - Y f(X, S) \Big| S=s\right] \\
    &= \argmin_{f\in\mathcal{F}_s}\mathbb{E}\left[\exp(f(X, S)) - \exp(f^*(X, S))f(X, S) \Big| S=s\right],
\end{align*}
where we used the fact that $\mathbb{E}[Y|X,S] = \exp(f^*(X,S))$. Minimizing the integrand pointwise with respect to $f(X,S)$ yields $f(X,S) = f^*(X,S)$.

\textbf{Squared Loss Regression.}
For the $l_2$ loss in Example \ref{ex4}, the risk minimization simplifies to:
\begin{align*}
    \argmin_{f\in\mathcal{F}_s}\mathbb{E}[(Y - f(X,S))^2|S=s]
    = \argmin_{f\in\mathcal{F}_s}\mathbb{E}\left[(f(X,S)- f^*(X,S))^2 + \text{Var}(\varepsilon) \Big| S=s\right].
\end{align*}
The solution is clearly $f(X,S) = f^*(X,S)$, the conditional expectation.

\textbf{Least Absolute Deviation (LAD) Regression.}
For the $l_1$ loss in Example \ref{ex-lad}, assuming the median of the noise $\varepsilon$ is 0, we have:
\begin{align*}
    \argmin_{f\in\mathcal{F}_s}\mathbb{E}[|Y - f(X,S)| | S=s]
    = \argmin_{f\in\mathcal{F}_s}\mathbb{E}[|f^*(X,S) + \varepsilon - f(X,S)| | S=s].
\end{align*}
Conditioning on $X=x$, the minimizer of $\mathbb{E}[|c + \varepsilon - z|]$ with respect to $z$ is the median of the distribution of $c+\varepsilon$, which is $c$ when $\text{Median}(\varepsilon)=0$. Thus, the minimizer is $f^*(X,S)$.

\subsection{Numerical implementation of Step 2}
The optimization in Step 2 \eqref{eq-est-Q} requires minimizing the empirical loss over a class of functions $\widetilde{\mathcal{Q}}$. Since $\widetilde{Q}$ is a quantile function, it is natural to restrict $\widetilde{\mathcal{Q}}$ to the set of non-decreasing functions. We discuss two efficient algorithms for this task.

\textbf{Isotonic Regression.}
If we make no assumptions on smoothness and simply require monotonicity, the estimation of $\widehat{Q}$ reduces to an \textit{isotonic regression} problem. Let the indices $(1), \dots, (n)$ permute the data such that the inputs are sorted: $\widehat{U}_{(1)} \le \widehat{U}_{(2)} \le \dots \le \widehat{U}_{(n)}$. Let $Y_{(i)}$ denote the response corresponding to $\widehat{U}_{(i)}$.
We aim to estimate the values $q_i = \widehat{Q}(\widehat{U}_{(i)})$ directly by solving:
\begin{align*}
    \min_{q_1, \dots, q_n} \sum_{i=1}^n \mathcal{L}(q_i, Y_{(i)})
    \quad \text{s.t.}\quad
    q_1 \le q_2 \le \cdots \le q_n.
\end{align*}
This constrained optimization can be solved efficiently using the Pool-Adjacent-Violators Algorithm (PAVA). Ties in the inputs (where $\widehat{U}_{(i)} = \widehat{U}_{(j)}$) are handled by enforcing equality constraints $q_i = q_j$ on the estimators. The resulting solution defines $\widehat{Q}$ as a non-decreasing step function with values $q_i$ at the observed points $\widehat{U}_{(i)}$.

\textbf{I-Spline Regression.}
If smoothness is desired, we can define $\widetilde{\mathcal{Q}}$ using Monotone Regression Splines (I-splines) \citep{ramsay1988monotone}. The I-spline basis functions, denoted by $\{\psi_j\}_{j=1}^M$, are defined on the interval $[0,1]$ as the cumulative integrals of M-splines (normalized B-splines). Because M-splines are non-negative, their integrals $\psi_j$ are inherently monotonically non-decreasing from 0 to 1. We model $Q$ as a linear combination of these basis functions:
$$ Q(u) = \alpha_0 + \sum_{j=1}^M \alpha_j \psi_j(u), \quad u \in [0,1], $$
where the complexity is determined by the number of knots and degree \citep{wang2021shape}. The estimation problem becomes:
\begin{align*}
    \min_{\alpha_0\in \mathbb{R},\ \alpha_j\geq 0} \sum_{i=1}^n \mathcal{L}\left(\alpha_0 +\sum_{j=1}^M \alpha_j\psi_j(\widehat{U}_i), Y_i \right).
\end{align*}
Since the basis functions $\psi_j$ are non-decreasing, constraining the coefficients to be non-negative ($\alpha_j \ge 0$) guarantees that the estimated function $\widehat{Q}$ is monotonically non-decreasing on $[0,1]$, ensuring it is a valid quantile function.

\section{Statistical theory}\label{sec:theory}
We begin by introducing the necessary notation. For $1\leq p<\infty$, let $L^p([0,1])$ denote the space of measurable functions $f$ defined on $[0,1]$ satisfying $\|f\|_p := \left(\int_{0}^{1}|f(u)|^p du\right)^{1/p}<\infty$. For any $Q_1, Q_2 \in L^2([0,1])$, we define the metric $d(Q_1, Q_2) = \|Q_1 - Q_2\|_2$. We denote the supremum norm by $\|f\|_\infty = \sup_{u\in [0,1]}|f(u)|$.
Let $\mathbb{P}_Y$ and $\mathbb{P}_{UY}$ denote the marginal distribution of $Y$ and the joint distribution of $(U, Y)$, respectively. For a measurable function $h$ of $(U,Y)$, we define the norm $\|h\|_{\mathbb{P}_{UY}}^2 = \mathbb{E}[h^2(U,Y)]$. Similarly, for a measurable function $g$ of $Y$, we define $\|g\|^2_{\mathbb{P}_Y} = \mathbb{E}[g^2(Y)]$. Finally, for any $a, b \in \mathbb{R}$, we write $a\vee b = \max\{a,b\}$.

We assume that the first-step estimator $\widehat{f}$ has been constructed using an independent dataset $\mathcal{D}^{est}$. For the estimation of $Q$, we consider a sample $\mathcal{D}=\{(X_i,S_i,Y_i):i=1,2,\dots,n\}$.
Let $\widehat{F}_s$ denote the empirical cumulative distribution function of $\{\widehat{f}(X_i,S_i): S_i=s\}$ within this sample:
\begin{align*}
\widehat{F}_s(\cdot)=\frac{1}{n_s}\sum_{i:S_i=s}\mathbf{1}\left\{\widehat{f}(X_i,S_i)\leq \cdot\right\},
\end{align*}
where $n_s=\sum^{n}_{i=1}\mathbf{1}\{S_i=s\}$. Throughout the analysis, we assume that a generic test point $(X, S)$ follows the same distribution as $(X_i, S_i)$ and is independent of the data used to estimate $\widehat{f}$.

To establish the consistency and convergence rates of the estimated fair prediction function, we restrict our analysis to the Sobolev space of order $k$ equipped with the supremum norm. Define:
$$\mathcal{W}^{k,\infty}:=\mathcal{W}^{k,\infty}([0,1])=\left\{Q: Q^{(a)}\in L^\infty([0,1]), \text{ for } 0\leq a\leq k\right\},$$
where the norm is given by $\|Q\|_{\mathcal{W}^{k,\infty}}=\max_{0\leq a\leq k}\|Q^{(a)}\|_\infty.$

\begin{assumption} \label{assum: Q}
Let $\mathcal{Q}=\{Q\in \mathcal{W}^{k,\infty}: \|Q\|_{\mathcal{W}^{k,\infty}}\leq R_1\}$ for some radius $R_1>0$ and smoothness $k\geq 2$. We assume that the true optimal quantile function satisfies $\widetilde{Q}\in \mathcal{Q}$.
\end{assumption}

Assumption~\ref{assum: Q} provides a regular function class for uniform control. We denote the true optimal fair prediction function by
$$f^*_{\text{fair}}(x,s):=\widetilde{Q}\circ F_{s}^*\circ f^*(x,s),$$
where $\widetilde{Q}$ is the optimal quantile function defined in Section \ref{sec:fair_para}.

\begin{assumption} \label{assum:loss}
	The loss function $\mathcal{L}(\cdot,\cdot)$ satisfies the following conditions:
    \begin{enumerate}[(i)]
        \item (Identifiability/Separability) For any $\varepsilon > 0$, there exists a constant $\lambda > 0$ such that:
	\begin{align}\label{eq-iden}
		\inf_{Q \in \mathcal{Q}: d(Q, \widetilde{Q}) \ge \varepsilon} \mathbb{E} \left[ \mathcal{L}\left\{Q(U),Y\right\} - \mathcal{L}\left\{\widetilde{Q}(U),Y\right\} \right] \ge \lambda d^2(Q, \widetilde{Q}).
	\end{align}
        \item (Lipschitz Continuity) For all $x,x'$ and $y$,
	$\vert \mathcal{L}\left(x,y\right) - \mathcal{L} \left(x',y\right) \vert  \le L(y) \vert x-x' \vert$
	where $\mathbb{E} \left[ L(Y)^2 \right]  < \infty$.
    \end{enumerate}
\end{assumption}

Below, we provide a brief discussion regarding $\widetilde{Q}$ and the separability condition in \eqref{eq-iden}.
Recall from Proposition~\ref{prop:risk-decomp-pointwise} that for any $Q\in\mathcal{Q}$,
\[
\mathbb{E}\!\left[\mathcal{L}(Q(U),Y)\right]
=\mathbb{E}_{U}\!\left[h\!\left(Q(U);Q_0^*(U),Q_1^*(U)\right)\right],
\qquad
h(q;y_0,y_1):=\sum_{s=0}^{1} r_s\,\mathcal{C}(y_s,q),
\]
where recall that $U$ is the probability integral transform defined in Section~\ref{sec:fair_para} with $U\sim \mathrm{Unif}[0,1]$. Thus, Assumption~\ref{assum:loss}(i) holds whenever $h(\cdot;y_0,y_1)$ is uniformly strongly convex in its decision
variable $q$, as formalized below.

\begin{proposition}\label{prop-sep-con}
Define $h(q;y_0,y_1):=\sum^{1}_{s=0}r_s \mathcal{C}(y_s,q)$. If $h(\cdot;y_0,y_1)$ is strongly convex with respect to $q$,
such that its second derivative with respect to $q$, denoted by $h^{(2)}(q;y_0,y_1)$, satisfies
$h^{(2)}(q;y_0,y_1)\geq 2\lambda$ for all $q,y_0,y_1$, then the separability condition in Assumption~\ref{assum:loss}(i) holds.
This condition is naturally satisfied for the squared loss and the cross-entropy loss.
\end{proposition}

\begin{assumption} \label{assum: f}
The first-step estimator $\widehat f(\cdot,s)$ satisfies a tail bound condition. Specifically, there exist constants $C>0$ and a sequence $r_{n,d}\rightarrow 0$ such that for any $\tau>0$ and $s \in \{0,1\}$:
$$\mathbb{P}(|\widehat f(X,S) - f^{\ast}(X,S)|>\tau|S=s)\leq C\exp(-C\tau^2/r_{n,d}^2).$$
\end{assumption}

\begin{assumption} \label{assum: F}
The CDF $F_s^{\ast}$ admits a bounded density for $s=0,1$. Furthermore, the sample proportion of the protected group converges in probability: $n_0/n\rightarrow^p r_0\in (0,1)$, where $n_0=\sum^{n}_{i=1}\mathbf{1}\{S_i=0\}$.
\end{assumption}

Assumption~\ref{assum: f} quantifies the first-step prediction error via a high-probability bound, and Assumption~\ref{assum: F} ensures smoothness/invertibility of the relevant distributional maps and non-degenerate group proportions.

\begin{theorem}[Consistency]\label{thm_consist}
Under Assumptions \ref{assum: Q}--\ref{assum: F}, as $n \to \infty$,
$$d(\widehat Q, \widetilde{Q})=o_p(1).$$
\end{theorem}

\begin{theorem}[Rate of convergence] \label{thm_rate}
Suppose Assumptions \ref{assum: Q}--\ref{assum: F} hold. Additionally, assume that the loss function's derivatives are bounded, i.e., $L(Y)\vee\sup_{|x|\leq R_2}|\mathcal{L}'(x,Y)|\vee |\mathcal{L}''(x,Y)|\leq B$ almost surely for some constant $B>0$. Then as $n \to \infty$,
\begin{align*}
 d(\widehat Q, \widetilde{Q}) = O_p\left(\{r_{n,d}\sqrt{\log(n)}\}^{\frac{2k+1}{2k+3}}\vee n^{-\frac{2k+1}{2(2k+3)}}\right).
\end{align*}
\end{theorem}

\begin{remark}\label{rm1}
{\rm
If the first-step estimator $\widehat f(\cdot,s)$ satisfies the moment condition:
	\begin{align}\label{as-moment}
		\sup_{s=0,1} \mathbb{E}\left[\left\vert
		\widehat f(X,S) - f^{\ast}(X,S) \right\vert^{2m} \Big|S=s\right]= O_p( r_{n,d}^{2m})
	\end{align}
    for $r_{n,d}\rightarrow 0$ and some $m\geq 1$, then by Lemma \ref{lem-uhat} and a slight modification of the proof of Theorem \ref{thm_rate}, the convergence rate becomes:
    \begin{align*}
    d(\widehat Q, \widetilde{Q}) = O_p\left(r_{n,d}^{\frac{m(2k+1)}{(m+1)(2k+3)}}\vee n^{-\frac{2k+1}{2(2k+3)}}\right).
    \end{align*}
}
\end{remark}

Finally, we translate the convergence results for the quantile function $\widehat{Q}$ into a bound on the estimation error of the final fair predictor. Define the $L_2$ estimation error metric:
$$
d(\widehat{f}_{\text{fair}}, f^*_{\text{fair}}):= \sqrt{\mathbb{E}_{X, S}\left[\left(\widehat{f}_{\text{fair}}(X, S) - f^*_{\text{fair}}(X , S) \right)^2\right]}.
$$

\begin{theorem}[Convergence Rate of the Fair Predictor]\label{thm_ffair}
Suppose Assumptions \ref{assum: Q}--\ref{assum: F} hold. Additionally, assume that the loss function derivatives are bounded almost surely, i.e., $L(Y)\vee\sup_{|x|\leq R_2}|\mathcal{L}'(x,Y)|\vee |\mathcal{L}''(x,Y)|\leq B$ for some constant $B>0$. Then, as $n \to \infty$:
$$
d(\widehat{f}_{\text{fair}}, f^*_{\text{fair}}) = O_p\left(\{r_{n,d}\sqrt{\log(n)}\}^{\frac{2k+1}{2k+3}}\vee n^{-\frac{2k+1}{2(2k+3)}}\right).
$$
Furthermore, under the moment conditions specified in Remark \ref{rm1}, the rate is
$$
d(\widehat{f}_{\text{fair}}, f^*_{\text{fair}}) = O_p\left(r_{n,d}^{\frac{m(2k+1)}{(m+1)(2k+3)}}\vee n^{-\frac{2k+1}{2(2k+3)}}\right).
$$
\end{theorem}

\section{Numerical experiments}\label{sec:exp}
We evaluate the empirical performance of our proposed framework using both real-world benchmark datasets and synthetic simulations. To demonstrate the versatility of our approach, we apply our method to three distinct regression tasks: quantile regression with pinball loss and Poisson regression with negative log-likelihood (Section \ref{sec:real_data}), followed by a robust regression simulation using Huber loss (Section \ref{sec:simulation}).
We remark that while our theoretical analysis assumes the latent predictor is estimated on a separate split to decouple estimation error from the fairness adjustment, in our numerical experiments we perform both estimation and calibration on the same training data. This approach is adopted to maximize sample efficiency, a common practice in finite-sample evaluation, though it technically introduces a dependency that is absent in the theoretical limit.

Throughout these studies, we quantify demographic disparity using the empirical Kolmogorov–Smirnov (KS) distance. For a predictor $\eta$ and a dataset $\{(X_i,S_i,Y_i)\}$, this metric is defined as:
\begin{equation*}
\operatorname{KS}_n(\eta)= \sup _{t \in \mathbb{R}}\left|\frac{1}{n_0} \sum_{i: S_i=0} \mathbf{1}\{\eta(X_i,S_i) \leq t\}-\frac{1}{n_1} \sum_{i:S_i=1} \mathbf{1}\{\eta(X_i,S_i) \leq t\}\right|,
\end{equation*}
where $n_s$ denotes the number of samples in the protected group $S=s$ for $s \in \{0,1\}$.

\subsection{Benchmark data applications}\label{sec:real_data}
\subsubsection{Quantile regression on CRIME dataset}
Pinball loss is the standard objective for quantile regression. It is defined as
$\mathcal{L}(\eta(X,S),Y)=\rho_\tau(Y-\eta(X,S)),$
where $\rho_\tau(u)=|u|(\tau\mathbf{1}\{u\geq0\}+ (1-\tau)\mathbf{1}\{u<0\})$, and $\tau \in (0,1)$ is the desired quantile level. To evaluate performance across the distribution, we conduct experiments setting $\tau \in \{0.25, 0.5, 0.75\}$.

\textbf{Implementation Details.}
We implement two versions of our approach corresponding to different choices of the function class $\widetilde{\mathcal{Q}}$ in the second step (Equation \eqref{eq-est-Q}):
\begin{enumerate}
    \item \textbf{Isotonic regression:} We set $\widetilde{\mathcal{Q}}$ as the class of all monotone functions on $[0,1]$ and solve the optimization using PAVA.
    \item \textbf{I-Spline regression:} We set $\widetilde{\mathcal{Q}}$ as the class of linear combinations of I-splines \citep{ramsay1988monotone} with non-negative coefficients, yielding smooth non-decreasing functions. We optimize the coefficients using the Hooke-Jeeves (HJ) algorithm \citep{moser2009hooke} since the pinball loss is non-differentiable.
\end{enumerate}
For the I-spline approach, we select the spline degree and number of interior knots via $K$-fold cross-validation (CV), with $K=5$, where $n_{\min}$ is the sample size of the smallest subgroup. We search candidate degrees in $\{1,\dots,10\}$ and interior knot counts in $\{0,\dots,10\}$ (spaced evenly in $(0,1)$). To select the optimal configuration, we first filter for degree-knot combinations yielding a KS distance within the lowest $10\%$ of all candidates, and among those, we select the one minimizing the regression risk (pinball loss).

We compare our method against standard Quantile Regression (QR) and the fair quantile regression approach proposed by \cite{liu2022conformalized}. The latter acts as a post-processing step that finds the fair function closest (in $L_2$ distance) to the original QR estimator. In contrast, our approach directly minimizes the pinball loss under the fairness constraint.

\textbf{Dataset.} We use the \textbf{Community \& Crime (CRIME)} dataset \citep{redmond2002data}, which contains 97 features describing socio-economic conditions of 1,994 U.S. communities. Following \cite{liu2022conformalized}, the objective is to predict the number of violent crimes per 100,000 population. The sensitive attribute is race, classified into four categories based on the predominant race in the community.

\textbf{Results.}
We split the data into 80\% training and 20\% testing sets, repeating the procedure 100 times. The mean and standard error of the empirical risk and KS distance are reported in Figure \ref{fig:crime3}.
Our method outperforms the competitor in terms of minimizing the pinball loss while maintaining low KS distance. This performance advantage stems from our direct optimization objective: we seek the risk minimizer under the constraint, whereas \cite{liu2022conformalized} optimize a proxy ($L_2$ distance to the unfair predictor). Furthermore, our method is more data-efficient as it estimates a single common quantile function, whereas the competing method requires estimating separate quantile functions for each sensitive group, which can be unstable for small subgroups. 

\begin{figure}[htbp]
  \centering
    \includegraphics[width=0.9\textwidth]{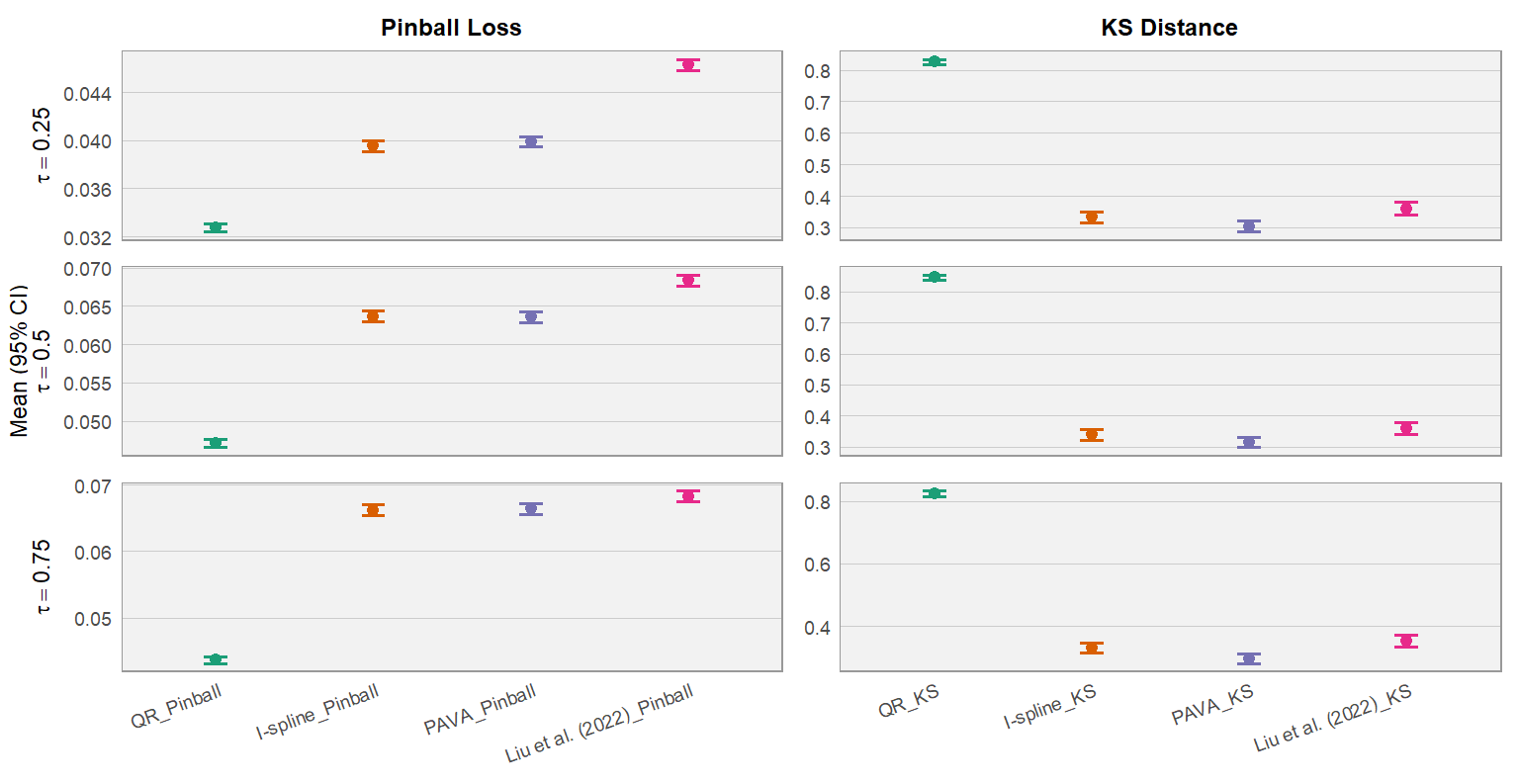}
  \caption{Pinball loss and KS distance comparison for $\tau=0.25,0.5,0.75$ on the CRIME dataset. $95\%$ confidence intervals for the mean are plotted. QR represents standard quantile regression.}
  \label{fig:crime3}
\end{figure}

%\begin{figure}[htbp]
%  \centering
%    \includegraphics[width=0.8\textwidth]{crime_0.25_k5.png}
%  \caption{Pinball loss and KS distance comparison for $\tau=0.25$ on the CRIME dataset. $95\%$ confidence intervals for the mean are plotted. QR represents standard quantile regression.}
%  \label{fig:crime0.25}
%\end{figure}

%\begin{figure}[htbp]
%  \centering
%    \includegraphics[width=0.8\textwidth]{crime_0.5_k5.png}
%  \caption{Pinball loss and KS distance comparison for $\tau=0.5$ on the CRIME dataset. $95\%$ confidence intervals for the mean are plotted.}
%  \label{fig:crime0.5}
%\end{figure}

%\begin{figure}[htbp]
%  \centering
%    \includegraphics[width=0.8\textwidth]{crime_0.75_k5.png}
%  \caption{Pinball loss and KS distance comparison for $\tau=0.75$ on the CRIME dataset. $95\%$ confidence intervals for the mean are plotted.}
%  \label{fig:crime0.75}
%\end{figure}

\subsubsection{Poisson Regression on HRS Dataset}
In Poisson regression, we minimize the negative log-likelihood: $\mathcal{L}(\eta(X,S),Y) = e^{\eta(X,S)} - Y\eta(X,S)$.
We apply the same two-step framework. For the PAVA implementation, we utilize a generalized version \citep{leeuw2009isotone} adapted for the Poisson loss. To address numerical instability where PAVA may produce negative infinity estimates when $Y_i=0$, we replace such values with the smallest finite estimate in the training set.

\textbf{Dataset.} We use the \textbf{Health \& Retirement Survey (HRS)} dataset, a longitudinal panel study from the University of Michigan that collects responses related to health and aging. The primary objective is to predict the number of dependencies individuals have in their daily activities, which is a count variable. We utilize the pre-processed version of the data from the \texttt{fairml} package \citep{fairmlpac}. After removing observations with missing data, the dataset comprises 12,766 instances. We utilize 26 patient-level predictors, including age, net worth, body mass index (BMI), and years of education. Marital status serves as the sensitive variable, categorized into two groups: `married/partnered' and `not married.'

\textbf{Results.}
Using the same 80/20 split over 100 repetitions, we compare our method against standard Poisson regression. As shown in Figure \ref{fig:hrsplot}, our approach achieves a superior balance of prediction accuracy (negative log-likelihood) and fairness (KS distance), demonstrating its effectiveness for count data.

\begin{figure}[htbp]
  \centering
    \includegraphics[width=0.8\textwidth]{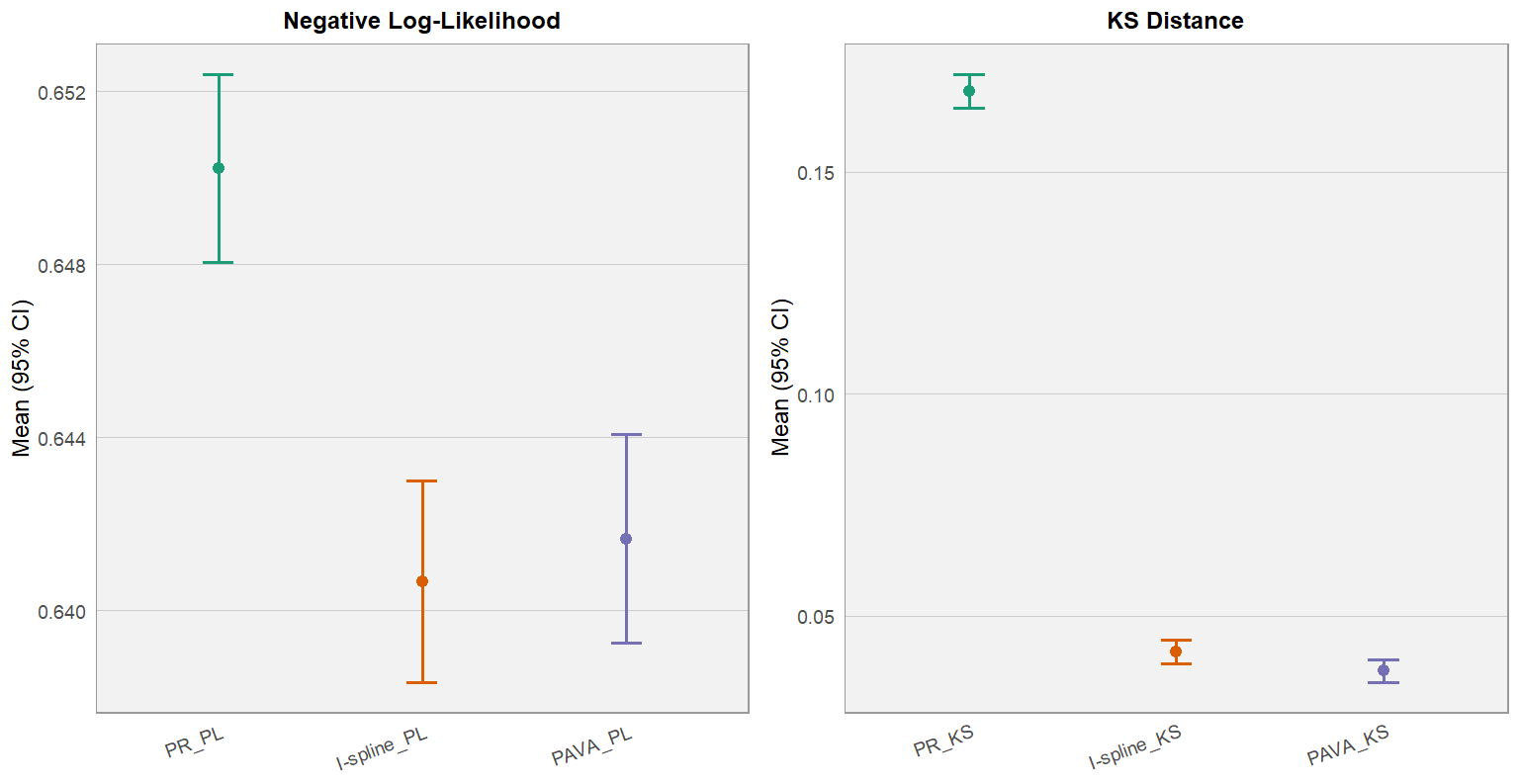}
  \caption{Comparison of methods on the Health \& Retirement Survey Dataset using Poisson regression. $95\%$ confidence intervals for the mean are plotted. PR denotes standard Poisson regression.}
  \label{fig:hrsplot}
\end{figure}

\subsection{Simulation Study: Robust Regression}\label{sec:simulation}

We further validate our algorithm using a robust regression task with the Huber loss. Following \cite{lambert2011robust}, we generate synthetic data according to:
\[
Y_i \;=\; 1 \;+\; X_i^\top \beta \;+\; a S_i \;+\; \sigma \,\varepsilon_i,
\]
where $X_i \sim \mathcal{N}_8(\mathbf{0}, \Sigma)$ with covariance $(\Sigma)_{ij}=0.5^{|i-j|}$. The coefficient vector is $\beta=(3,\;1.5,\;0,\;0,\;2,\;0,\;0,\;0)^\top$ and $\sigma=9.67$.
To introduce heavy tails and outliers, the error term $\varepsilon_i$ is drawn from a standardized Gaussian mixture:
\[
V_i \sim 0.9\,\mathcal{N}(0,1) \;+\; 0.1\,\mathcal{N}(0,225), \quad \varepsilon_i = \frac{V_i}{\sqrt{\mathrm{Var}(V_i)}}.
\]
The sensitive variable is $S_i\sim\mathrm{Bernoulli}(0.5)$. We test two levels of bias by setting $a=1$ and $a=5$.
The training set size is $n=500$, and the test set contains 10,000 observations (balanced between $S=0$ and $S=1$).
We minimize the Huber loss with $M = 1.345\sigma \approx 13.01$, a choice that balances robustness and efficiency for Gaussian errors \citep{lambert2011robust, huber1981robust}.

\textbf{Results.}
We apply our two-step procedure (using both isotonic regression and I-spline with CV) and compare it to standard robust regression (RR). The results over 200 repetitions are displayed in Figure~\ref{fig:RRa=1} ($a=1$) and Figure~\ref{fig:RRa=5} ($a=5$).
Our method effectively reduces unfairness regardless of the initial disparity magnitude. Notably, as the inherent unfairness in the data increases (larger $a$), the trade-off becomes more pronounced: the fair regression risk deviates more significantly from the unconstrained risk, reflecting the necessary cost of enforcing parity in a highly biased system.

\begin{figure}[htbp]
  \centering
    \includegraphics[width=0.8\textwidth]{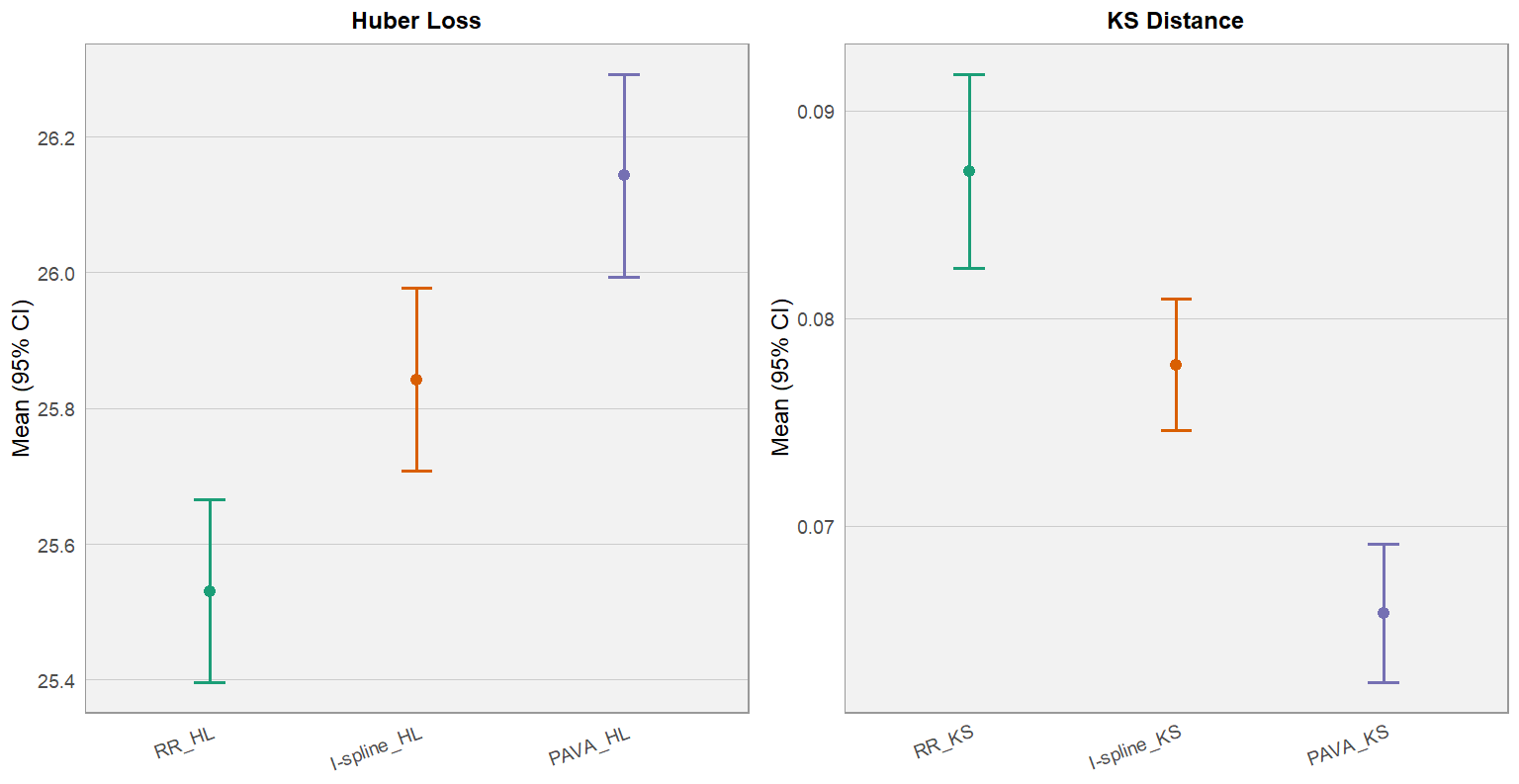}
  \caption{Huber loss and KS distance for robust regression simulation with moderate bias ($a = 1$). Approximated $95\%$ CIs are plotted. RR denotes standard robust regression.}
  \label{fig:RRa=1}
\end{figure}

\begin{figure}[htbp]
  \centering
    \includegraphics[width=0.8\textwidth]{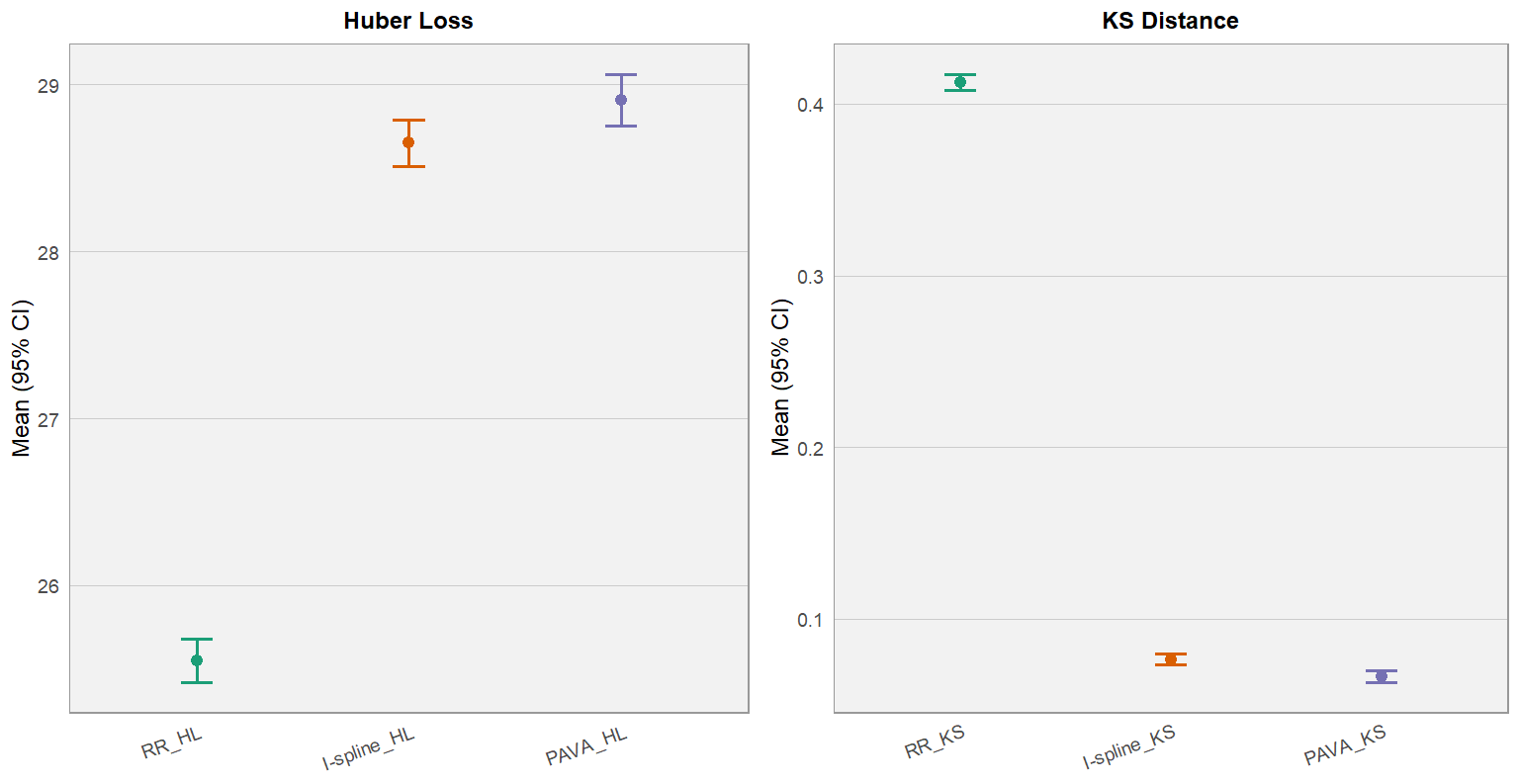}
  \caption{Huber loss and KS distance for robust regression simulation with high bias ($a = 5$). Approximated $95\%$ CIs are plotted. RR denotes standard robust regression.}
  \label{fig:RRa=5}
\end{figure}

\section{Discussion}\label{sec:discuss}
In this paper, we introduced a unified framework for fair regression under demographic parity, leveraging a novel characterization of the optimal fair risk minimizer. We conclude by highlighting several promising avenues for future research.

First, our current theoretical framework assumes that the latent regression function $f^*$ is scalar-valued. However, many complex applications involve multivariate latent structures. For instance, consider a location-scale model $Y=g^*(X,S)+\sigma^*(X,S)\varepsilon$, where the latent parameter is the vector $f^*=(g^*,\sigma^*) \in \mathbb{R}^2$. In such settings, where both the conditional mean and variance depend on sensitive attributes, the univariate characterization in Proposition \ref{prop2} does not directly apply. Developing alternative characterizations for these multivariate settings remains an open and non-trivial problem.

Second, while this work focuses on demographic parity, it would be of significant interest to investigate whether the connection between fair regression and the Kantorovich barycenter problem extends to other fairness criteria. Adapting our optimal transport perspective to constraints such as equal opportunity or equalized odds could provide similar unified frameworks for a broader class of fair learning tasks.

\bibliographystyle{chicago}
\bibliography{references}

@book{van1996weak,
  title={Weak Convergence and Empirical Processes},
  author={van der Vaart, Aad W and Wellner, Jon A},
  year={1996},
  publisher={Springer}
}

@book{geer2000empirical,
  title={Empirical Processes in M-estimation},
  author={van de Geer, Sara A},
  volume={6},
  year={2000},
  publisher={Cambridge university press}
}

@article{mehrabi2021survey,
  title={A survey on bias and fairness in machine learning},
  author={Mehrabi, Ninareh and Morstatter, Fred and Saxena, Nripsuta and Lerman, Kristina and Galstyan, Aram},
  journal={ACM computing surveys (CSUR)},
  volume={54},
  number={6},
  pages={1--35},
  year={2021},
  publisher={ACM New York, NY, USA}
}

@article{caton2024fairness,
  title={Fairness in machine learning: A survey},
  author={Caton, Simon and Haas, Christian},
  journal={ACM Computing Surveys},
  volume={56},
  number={7},
  pages={1--38},
  year={2024},
  publisher={ACM New York, NY}
}

@article{pessach2022review,
  title={A review on fairness in machine learning},
  author={Pessach, Dana and Shmueli, Erez},
  journal={ACM Computing Surveys (CSUR)},
  volume={55},
  number={3},
  pages={1--44},
  year={2022},
  publisher={ACM New York, NY}
}

@article{agueh2011barycenters,
  title={Barycenters in the {Wasserstein} space},
  author={Agueh, Martial and Carlier, Guillaume},
  journal={SIAM Journal on Mathematical Analysis},
  volume={43},
  number={2},
  pages={904--924},
  year={2011},
  publisher={SIAM}
}

@article{santambrogio2015optimal,
  title={Optimal transport for applied mathematicians},
  author={Santambrogio, Filippo},
  journal={Birk{\"a}user, NY},
  volume={55},
  number={58-63},
  pages={94},
  year={2015},
  publisher={Springer}
}

@inproceedings{pratelli2007equality,
  title={On the equality between {Monge's} infimum and {Kantorovich's} minimum in optimal mass transportation},
  author={Pratelli, Aldo},
  booktitle={Annales de l'Institut Henri Poincare (B) Probability and Statistics},
  volume={43},
  number={1},
  pages={1--13},
  year={2007},
  organization={Elsevier}
}

@inproceedings{gaucher2023fair,
  title={Fair learning with {Wasserstein} barycenters for non-decomposable performance measures},
  author={Gaucher, Solenne and Schreuder, Nicolas and Chzhen, Evgenii},
  booktitle={International Conference on Artificial Intelligence and Statistics},
  pages={2436--2459},
  year={2023},
  organization={PMLR}
}

@article{chzhen2020fair,
  title={Fair regression with {Wasserstein} barycenters},
  author={Chzhen, Evgenii and Denis, Christophe and Hebiri, Mohamed and Oneto, Luca and Pontil, Massimiliano},
  journal={Advances in Neural Information Processing Systems},
  volume={33},
  pages={7321--7331},
  year={2020}
}

@article{liu2022conformalized,
  title={Conformalized Fairness via Quantile Regression},
  author={Liu, Meichen and Ding, Lei and Yu, Dengdeng and Liu, Wulong and Kong, Linglong and Jiang, Bei},
  journal={Advances in Neural Information Processing Systems},
  volume={35},
  pages={11561--11572},
  year={2022}
}

@inproceedings{do2022fair,
  title={Fair generalized linear models with a convex penalty},
  author={Do, Hyungrok and Putzel, Preston and Martin, Axel S and Smyth, Padhraic and Zhong, Judy},
  booktitle={International Conference on Machine Learning},
  pages={5286--5308},
  year={2022},
  organization={PMLR}
}

@inproceedings{agarwal2019fair,
  title={Fair regression: Quantitative definitions and reduction-based algorithms},
  author={Agarwal, Alekh and Dud{\'\i}k, Miroslav and Wu, Zhiwei Steven},
  booktitle={International Conference on Machine Learning},
  pages={120--129},
  year={2019},
  organization={PMLR}
}

@inproceedings{agarwal2018reductions,
  title={A reductions approach to fair classification},
  author={Agarwal, Alekh and Beygelzimer, Alina and Dud{\'\i}k, Miroslav and Langford, John and Wallach, Hanna},
  booktitle={International conference on machine learning},
  pages={60--69},
  year={2018},
  organization={PMLR}
}

@inproceedings{silvia2020general,
  title={A general approach to fairness with optimal transport},
  author={Silvia, Chiappa and Ray, Jiang and Tom, Stepleton and Aldo, Pacchiano and Heinrich, Jiang and John, Aslanides},
  booktitle={Proceedings of the AAAI Conference on Artificial Intelligence},
  volume={34},
  number={04},
  pages={3633--3640},
  year={2020}
}

@article{redmond2002data,
  title={A data-driven software tool for enabling cooperative information sharing among police departments},
  author={Redmond, Michael and Baveja, Alok},
  journal={European Journal of Operational Research},
  volume={141},
  number={3},
  pages={660--678},
  year={2002},
  publisher={Elsevier}
}

@inproceedings{komiyama2018nonconvex,
  title={Nonconvex optimization for regression with fairness constraints},
  author={Komiyama, Junpei and Takeda, Akiko and Honda, Junya and Shimao, Hajime},
  booktitle={International conference on machine learning},
  pages={2737--2746},
  year={2018},
  organization={PMLR}
}

@inproceedings{wu2019convexity,
  title={On convexity and bounds of fairness-aware classification},
  author={Wu, Yongkai and Zhang, Lu and Wu, Xintao},
  booktitle={The World Wide Web Conference},
  pages={3356--3362},
  year={2019}
}

@article{leeuw2009isotone,
  title={Isotone optimization in {R}: pool-adjacent-violators algorithm ({PAVA}) and active set methods},
  author={Leeuw, Jan de and Hornik, Kurt and Mair, Patrick},
  year={2009}
}

@inproceedings{moser2009hooke,
  title={Hooke-jeeves revisited},
  author={Moser, Irene},
  booktitle={2009 IEEE Congress on Evolutionary Computation},
  pages={2670--2676},
  year={2009},
  organization={IEEE}
}

@article{ramsay1988monotone,
  title={Monotone regression splines in action},
  author={Ramsay, James O},
  journal={Statistical science},
  pages={425--441},
  year={1988},
  publisher={JSTOR}
}

@inproceedings{Donini-Oneto-Ben-David18,
	author = {Donini, Michele and Oneto, Luca and Ben-David, Shai and Shawe-Taylor, John S and Pontil, Massimiliano},
	booktitle = {Advances in Neural Information Processing Systems},
	pages = {2796--2806},
	title = {Empirical risk minimization under fairness constraints},
	volume = {31},
	year = {2018}}

@Manual{fairmlpac,
title = {fairml: Fair Models in Machine Learning},
author = {Marco Scutari},
year = {2023},
note = {R package version 0.8},
url = {https://CRAN.R-project.org/package=fairml},
  }

@article{lambert2011robust,
  title={Robust regression through the Huber’s criterion and adaptive lasso penalty},
  author={Lambert-Lacroix, Sophie and Zwald, Laurent},
  year={2011}
}

@article{huber1981robust,
  title={Robust Statistics, Wiler},
  author={Huber, PJ},
  journal={New York},
  year={1981}
}

@article{wang2021shape,
  title={Shape-Restricted Regression Splines with R Package splines2.},
  author={Wang, Wenjie and Yan, Jun},
  journal={Journal of Data Science},
  volume={19},
  number={3},
  year={2021}
}

@article{li2024multi,
  title={Multi-Output Distributional Fairness via Post-Processing},
  author={Li, Gang and Lin, Qihang and Ghosh, Ayush and Yang, Tianbao},
  journal={arXiv preprint arXiv:2409.00553},
  year={2024}
}

@article{fallah2025statistical,
  title={The Statistical Fairness-Accuracy Frontier},
  author={Fallah, Alireza and Jordan, Michael I and Ulichney, Annie},
  journal={arXiv preprint arXiv:2508.17622},
  year={2025}
}

@article{fukuchi2025meta,
  title={Meta Optimality for Demographic Parity Constrained Regression via Post-Processing},
  author={Fukuchi, Kazuto},
  journal={arXiv preprint arXiv:2506.13947},
  year={2025}
}

@inproceedings{hu2023fairness,
  title={Fairness in multi-task learning via wasserstein barycenters},
  author={Hu, Fran{\c{c}}ois and Ratz, Philipp and Charpentier, Arthur},
  booktitle={Joint European Conference on Machine Learning and Knowledge Discovery in Databases},
  pages={295--312},
  year={2023},
  organization={Springer}
}

@article{taturyan2024regression,
  title={Regression under demographic parity constraints via unlabeled post-processing},
  author={Taturyan, Gayane and Chzhen, Evgenii and Hebiri, Mohamed},
  journal={Advances in Neural Information Processing Systems},
  volume={37},
  pages={117917--117953},
  year={2024}
}

@article{fukuchi2023demographic,
  title={Demographic parity constrained minimax optimal regression under linear model},
  author={Fukuchi, Kazuto and Sakuma, Jun},
  journal={Advances in Neural Information Processing Systems},
  volume={36},
  pages={8653--8689},
  year={2023}
}

@article{carrizosa2025fair,
  title={Fair Logistic Regression via Penalized Mean-Field Variational Inference},
  author={Carrizosa, Emilio and Jim{\'e}nez-Llamas, Rafael and Ram{\i}rez-Cobo, Pepa},
  journal={Pre-print. URL: https://shorturl. at/by0UY},
  year={2025}
}

@article{deza2026fair,
  title={Fair and accurate regression: Strong formulations and algorithms},
  author={Deza, Anna and G{\'o}mez, Andr{\'e}s and Atamt{\"u}rk, Alper},
  journal={INFORMS Journal on Optimization},
  year={2026},
  publisher={INFORMS}
}

%%%%%%%%%%%%%%%%%%%%%%%%%%%%%%%%%%%%%%%%%%%%%%%%%%%%%%%%%%%%

%\newpage

\appendix

\section{Appendix} 
\subsection{Technical details}\label{sec:proof}

\begin{proof}[Proof of Proposition \ref{prop1}]
We first note that
	\begin{align*}
	F^*:=&\min_{f: f(X^0,0)\overset{d}{=} f(X^1,1)}\mathbb{E}[\mathcal{L}(f(X,S),Y)]\\
 =&\min_{f: f(X^0,0)\overset{d}{=} f(X^1,1)}\mathbb{E}[\mathbb{E}[\mathcal{L}(f(X,S),Y)|X,S]]
		\\=&\min_{f: f(X^0,0)\overset{d}{=} f(X^1,1)}\mathbb{E}[\mathcal{C}(f^*(X,S),f(X,S))]
		\\=&\min_{f: f(X^0,0)\overset{d}{=} f(X^1,1)}\sum^{1}_{s=0}r_s\mathbb{E}[\mathcal{C}(f^*(X,S),f(X,S))|S=s].
	\end{align*}
Recall the KBP defined in (\ref{eq-kbp}) and denote by $\bar{f}$ and $\bar{\nu}$ the solutions to the FRP and KBP, respectively. Clearly, we have
	\begin{align*}
		K^*\leq  \sum_{s=0}^1 r_s\inf_{\gamma\sim \Pi(\nu_{f^*|s},\nu_{\bar{f}}) }\mathbb{E}_{(U^*,U)\sim \gamma}[\mathcal{C}(U^*,U)] \leq \sum^{1}_{s=0}r_s\mathbb{E}[\mathcal{C}(f^*(X,S),\bar{f}(X,S))|S=s]=F^*,  
	\end{align*}
	where $\nu_{\bar{f}}$ represents the distribution of $\bar{f}(X,S)$ which is independent of $S$ by the fairness constraint. On the other hand, due to the existence of the transport maps, we have
	\begin{align*}
		K^*=&\sum_{s=0}^1 r_s\inf_{\gamma\sim \Pi(\nu_{f^*|s},\bar{\nu}) }\mathbb{E}_{(U^*,U)\sim \gamma}[\mathcal{C}(U^*,U)] 
		=\sum_{s=0}^1 r_s\mathbb{E}[\mathcal{C}(f^*(X,S),T_S\circ f^*(X,S))|S=s], 
	\end{align*}
	where $T_S\circ f^*(X,S) \sim \bar{\nu}$, which does not depend on the value of $S$. As $T_S\circ f^*(X,S)$ is fair, we must have 
	$$K^*=\sum_{s=0}^1 r_s\mathbb{E}[\mathcal{C}(f^*(X,S),T_S\circ f^*(X,S))|S=s]\geq F^*.$$ Therefore, we have
$K^*=F^*$ and $\bar{f}(X,S)=T_{S}\circ f^*(X,S)$. 
\end{proof}

\begin{proof}[Proof of Proposition \ref{prop2}]
	{\rm 
	 By Theorem 1.38 of \cite{santambrogio2015optimal}, the support of any optimal solution $\gamma$ to the Kantorovich problem is a $\mathcal{C}$-cyclically monotone set. This means that for any $(x,y),(x',y')$ belonging to the support of $\gamma$, we have
		\begin{align}\label{eq-c-1}
			\mathcal{C}(x,y)+\mathcal{C}(x',y')\leq \mathcal{C}(x,y')+\mathcal{C}(x',y).    
		\end{align}
		Based on this, we aim to show that $x<x'$ implies $y\leq y'$, which will allow us to conclude from Lemma 2.8 of \cite{santambrogio2015optimal}. We prove this using contradiction. 
		Let $\delta_x=x'-x$ and $\delta_y=y-y'$. Assuming that $y>y'$, we have $\delta_x,\delta_y>0.$ Then, (\ref{eq-c-1}) can be rewritten as
		\begin{align*}
			\mathcal{C}(x,y)+\mathcal{C}(x+\delta_x,y-\delta_y)-\mathcal{C}(x,y-\delta_y)-\mathcal{C}(x+\delta_x,y)\leq 0.    
		\end{align*}
		Let $c$ be the boundary of the rectangle with the vertices $(x,y-\delta_y),(x+\delta_x,y-\delta_y),(x+\delta_x,y)$ and $(x,y)$ (the curve $c$ is oriented counterclockwise). Note that
		\begin{align*}
			\oint_c \frac{\partial \mathcal{C}(u,v)}{\partial u} du
			=&\int^{x+\delta_x}_{x}\frac{\partial \mathcal{C}(u,y-\delta_y)}{\partial u} du-\int^{x+\delta_x}_{x}\frac{\partial \mathcal{C}(u,y)}{\partial u} du
			\\=&\mathcal{C}(x+\delta_x,y-\delta_y)-\mathcal{C}(x,y-\delta_y)-\{\mathcal{C}(x+\delta_x,y)-\mathcal{C}(x,y)\}\leq 0.
		\end{align*}
		By Green's Theorem, 
		\begin{align*}
			\oint_c \frac{\partial \mathcal{C}(u,v)}{\partial u} du=-\int\int_R \frac{\partial^2 \mathcal{C}(u,v)}{\partial u\partial v} dudv> 0,    
		\end{align*}
		where $R$ denotes the rectangular region defined above with the boundary $c$ and we have used Condition (\ref{con-twice}) to get the inequality. The contradiction suggests that $\delta_y\leq 0$, i.e., $y\leq y'.$
	}    
\end{proof}

\begin{proof}[Proof of Proposition \ref{prop:risk-decomp-pointwise}]
By the law of iterated expectation and the definition of $\mathcal{C}$ in \eqref{eq-cost},
\begin{align*}
\mathbb{E}\!\left[\mathcal{L}(Q(U),Y)\right]
&=\mathbb{E}\!\left[\mathbb{E}\!\left\{\mathcal{L}\!\left(Q\!\circ F_S^*\!\circ f^*(X,S),Y\right)\middle|X,S\right\}\right]\\
&=\mathbb{E}\!\left[\mathcal{C}\!\left(f^*(X,S),\,Q\!\circ F_S^*\!\circ f^*(X,S)\right)\right].
\end{align*}
Using $U:=F_S^*\circ f^*(X,S)$ and $f^*(X,S)=Q_S^*(U)$, this becomes
\[
\mathbb{E}\!\left[\mathcal{C}\!\left(Q_S^*(U),\,Q(U)\right)\right].
\]
Since $U\mid S=s\sim\mathrm{Unif}[0,1]$ for $s=0,1$, we have $U\perp S$ and hence $\mathbb{P}(S=s\mid U)=r_s$.
Therefore,
\[
\mathbb{E}\!\left[\mathcal{C}\!\left(Q_S^*(U),\,Q(U)\right)\right]
=\mathbb{E}_{U}\!\left[\sum_{s=0}^{1} r_s\,\mathcal{C}\!\left(Q_s^*(U),\,Q(U)\right)\right],
\]
which proves \eqref{eq:risk-decomp}.
For \eqref{eq:Q-pointwise}, define $h_u(q):=\sum_{s=0}^{1} r_s\,\mathcal{C}\!\left(Q_s^*(u),q\right)$ so that
\eqref{eq:risk-decomp} rewrites the population objective as $\mathbb{E}_{U}[h_U(Q(U))]$.
Since $h_u$ depends on $Q$ only through the single value $Q(u)$, any minimizer $\widetilde{Q}$ must minimize $h_u(\cdot)$
pointwise for a.e.\ $u$, yielding \eqref{eq:Q-pointwise}.
\end{proof}

\begin{lemma}\label{lem-aux-1}
For two sets of random variables $\{u_i\}^{n}_{i=1}$ and $\{v_i\}^{n}_{i=1}$, 
let $F_{n,u}(t)$ and $F_{n,v}(t)$ be the corresponding empirical distribution functions. We have 
\begin{align*}
\sup_t\left|F_{n,u}(t)-F_{n,v}(t)\right|\leq \frac{1}{ n}\sum^{n}_{i=1}\mathbf{1}\{|u_i-v_i|>\tau\} + \sup_t \{F_{n,u}(t+\tau) - F_{n,u}(t-\tau)\},    
\end{align*}
for any $\tau>0.$
\end{lemma}
\begin{proof}[Proof of Lemma \ref{lem-aux-1}]
Note that
\begin{align*}
 \left|F_{n,u}(t)-F_{n,v}(t)\right|=& \frac{1}{n}\left|\sum^{n}_{i=1}\left(\mathbf{1}\{u_i\leq t\}-\mathbf{1}\{v_i\leq t\}\right)\right|
 \\ \leq & \frac{1}{n}\sum^{n}_{i=1}\left(\mathbf{1}\{u_i\leq t, v_i>t\} + \mathbf{1}\{v_i\leq t, u_i>t\}\right)
 \\ \leq & \frac{1}{n}\sum^{n}_{i=1}\left(\mathbf{1}\{t-\tau < u_i\leq t, v_i>t\} + \mathbf{1}\{v_i\leq t, t<u_i\leq t+\tau\}\right)
 \\ & + \frac{1}{n}\sum^{n}_{i=1}\left(\mathbf{1}\{u_i\leq t-\tau, v_i>t\} + \mathbf{1}\{v_i\leq t, u_i>t+\tau\}\right)
\\ \leq & \frac{1}{n}\sum^{n}_{i=1}\mathbf{1}\{t-\tau<u_i\leq t+\tau\}+\frac{1}{n}\sum^{n}_{i=1}\mathbf{1}\{|u_i-v_i|>\tau\}.
\end{align*}
The conclusion follows by taking the supremum over $t$.
\end{proof}

\begin{proof}[Proof of Theorem \ref{thm_consist}]
Throughout the arguments below, we let $C$ be a generic constant that does not depend on $n$ and is allowed to vary from one instance to another. Define $M(Q) \equiv -\mE \left[ \cL \left\{Q(U),Y\right\} \right]$ and $M_n(Q) \equiv -n^{-1} \sum_{i=1}^{n} \cL \left\{Q(U_i),Y_i \right\}$. It follows that $\widetilde{Q} = \argmax_{Q \in \cQ} M(Q)$. By the separable condition in Assumption \ref{assum:loss}, for any $\varepsilon>0$, we have 
	\begin{align*}
		\mP(d(\widehat Q, \widetilde{Q})> \varepsilon) \le& \mP(M(\widetilde{Q}) - M(\widehat Q) > \delta)\\
		\le& \mP(M(\widetilde{Q}) - M_n(\widetilde{Q}) > \delta/3) + \mP(M_n(\widetilde{Q}) - M_n(\widehat Q) > \delta/3)\\
		&+ \mP(M_n(\widehat Q) - M(\widehat Q) > \delta/3)
		\\ :=& I_1 + I_2 + I_3
	\end{align*}
with $\delta=\lambda\epsilon^2.$ For $Q_1,Q_2 \in \cQ$, Assumption \ref{assum:loss} (ii) entails that 
$$\left\vert \cL(Q_1(x),y) - \cL(Q_2(x),y)) \right\vert \le \| Q_1 - Q_2 \|_\infty | L(y)|.$$
By Theorem 2.4 of \cite{geer2000empirical}, we have
\begin{align*}
N\left(\varepsilon,\cQ,\|\cdot\|_\infty\right)\leq A \varepsilon^{-1/k}
\end{align*}
for all $\varepsilon>0$ and some $A>0.$ Let $Q_1,\dots,Q_{N_\varepsilon}$ be an $\varepsilon$-cover of $\cQ$ under the metric $\|\cdot\|_\infty$, i.e., for any $Q\in \cQ$, there exists a $Q_i$ such that $\|Q-Q_i\|_\infty \leq \varepsilon$. Then the brackets $[\cL(Q_i(\cdot),\cdot)-\varepsilon L, \cL(Q_i(\cdot),\cdot)+\varepsilon L]$ for $1\leq i\leq N_\varepsilon$ form a cover of $\mathcal M = \{ \cL(Q(\cdot),\cdot), Q \in \cQ \}$ under the metric $\|\cdot\|_{\mP_{UY}}$ because
\begin{align*}
\vert \cL(Q(\cdot),\cdot)-\cL(Q_i(\cdot),\cdot) \vert  &\le L(y) \|Q-Q_i\|_\infty\\   
&\le \varepsilon L(y), 
\end{align*}
where the first inequality comes from Assumption \ref{assum:loss} (ii).

%\begin{align*}
%\|\cL(Q(\cdot),\cdot)-\cL(Q_i(\cdot),\cdot)\|_{\mP_{UY}}^2:=&\mE[\{\cL(Q(U),Y) - \cL(Q_i(U),Y)\}^2]
%\\ \leq& \| Q - Q_i \|_\infty^2 \mE[L^2(Y)]
%\\ \leq& \varepsilon^2 \|L\|_{\mP_Y}^2.
%\end{align*}
As a result, we have
\begin{align*}
N_{[\,]} \left(2\varepsilon \|L\|_{\mP_Y},\mathcal M, \|\cdot\|_{\mP_{UY}}\right) \le N\left(\varepsilon, \cQ, \Vert \cdot \Vert_{\infty}\right) \le A \varepsilon^{-1/k},
\end{align*}
which implies that $\mathcal M$ is a Glivenko-Cantelli class, i.e., $\sup_{Q \in \cQ} \left\vert M_n(Q) - M(Q) \right\vert =o_p(1)$. Thus, both $I_1$ and $I_3$ converge to zero as $n\rightarrow \infty$. 
%{\color{blue} $N_{[\,]} \left(2\varepsilon \|L\|_{\mP_Y},\mathcal M, \|\cdot\|_{\mP_{UY}}\right)$ is not a bracketing number. For example, in the squared loss case, $ \cL(Q_i(U), Y) -\epsilon \Vert L\Vert_{\mP_Y}  - \cL(Q(U), Y) = (2Y - Q_i(U) - Q(U))(Q(U) - Q_i(U))  -\epsilon \Vert L\Vert_{\mP_Y}$, which may not be negative for some $Q$ that $\Vert Q-Q_i\Vert_\infty \leq \epsilon$.  Instead, we can use brackets as $[\cL(Q_i(\cdot),\cdot)-\varepsilon L(\cdot), \cL(Q_i(\cdot),\cdot)+\varepsilon L(\cdot)]$, since $|\cL(Q(U), Y)  - \cL(Q_i(U), Y)|\leq L(Y)\|Q - Q_i\|_{\infty}\leq \varepsilon  L(Y)$. And $2\varepsilon\| L(Y)\|_{\mP_{UY}} = 2\varepsilon\| L(Y)\|_{\mP_{Y}}$.  }

Below, we show that $I_2\rightarrow 0$ as $n\rightarrow \infty$. 
Using Assumptions \ref{assum: Q}-\ref{assum:loss}, we get
	\begin{align*}	\sup_{Q\in \cQ}\left|\sum_{i=1}^n\mathcal{L}(Q(\widehat{U}_i),Y_i)-\sum_{i=1}^n\mathcal{L}(Q(U_i),Y_i)\right|    
		\leq & \sum_{i=1}^n L(Y_i)\sup_{Q\in \cQ}|Q(\widehat{U}_i)-Q(U_i)|
		\\ \le& C \sum_{i=1}^{n} L(Y_i)\left\vert \widehat U_i - U_i \right\vert
		\\ = & C \sum_{i:S_i=0} L(Y_i)|\widehat{F}_{0}\circ \widehat{f}(X_i,0)-F_{0}^{\ast}\circ f^*(X_i,0)|
		\\&+C \sum_{i:S_i=1} L(Y_i)|\widehat{F}_{1}\circ \widehat{f}(X_i,1)-F_{1}^{\ast}\circ f^*(X_i,1)|,
	\end{align*}
where we have used the mean-value theorem and the fact that $\sup_{Q\in \cQ}\|Q^{(1)}\|_\infty\leq R_2$ to get $$\sup_{Q\in \cQ}|Q(x)-Q(x')|\leq \sup_{Q\in \cQ} \|Q^{(1)}\|_\infty |x-x'| \leq  R_2 |x-x'|.$$ 
Below, we deal with the term $\sum_{i:S_i=0} L(Y_i)|\widehat{F}_{0}\circ \widehat{f}(X_i,0)-F_{0}^{\ast}\circ f^*(X_i,0)|$ and remark that the other term can be handled in the same way. Note that
	\begin{align*}
		& \sum_{i:S_i=0} L(Y_i)|\widehat{F}_{0}\circ \widehat{f}(X_i,0)-F_{0}^{\ast}\circ f^*(X_i,0)|
		\\ \leq &  \sum_{i:S_i=0} L(Y_i)|\widehat{F}_{0}\circ \widehat{f}(X_i,0)-F_{0}^{\ast}\circ \widehat{f}(X_i,0)| +  \sum_{i:S_i=0} L(Y_i)| F_{0}^{\ast}\circ \widehat{f}(X_i,0)-F_{0}^{\ast}\circ f^*(X_i,0)|
		\\ :=&J_1 + J_2.
	\end{align*}
By the Cauchy–Schwarz inequality and Assumption \ref{assum: F}, we have 
\begin{align*}
		J_2\leq& C\sum_{i:S_i=0} L(Y_i)| \widehat{f}(X_i,0)-f^*(X_i,0)|
		\\ \leq& C\left(\sum_{i:S_i=0} L(Y_i)^2\right)^{1/2}\left(\sum_{i:S_i=0}| \widehat{f}(X_i,0)-f^*(X_i,0)|^2\right)^{1/2}.
\end{align*}
The first term on the last inequality is of the order $O_p(n^{1/2})$. For the second term, we have 
	\begin{align*}
		&\mathbb{E}\left[\left(\sum_{i:S_i=0}| \widehat{f}(X_i,0)-f^*(X_i,0)|^2\right)^{1/2}\right]
		\\ \leq & \left(\sum_{i=1}^n \mathbb{E}[| \widehat{f}(X_i,0)-f^*(X_i,0)|^2|S_i=0]P(S_i=0)\right)^{1/2}\leq Cn^{1/2}r_{n,d}. 
	\end{align*}
Therefore, $J_2=O_p(nr_{n,d})$. On the other hand, let $\widetilde{F}_s$ be the empirical distribution function of $\mathcal{D}_s^*=\{f^*(X_i,S_i): S_i=s,1\leq i\leq n\}$. We have 
	\begin{align*}
		J_1 \leq \left(\sup_v|\widehat{F}_{0}(v)-\widetilde{F}_{0}(v)|+ \sup_v |\widetilde{F}_{0}(v)-F_{0}^*(v)|\right)\left(\sum_{i:S_i=0}L(Y_i)\right),    
	\end{align*}
	where $\sum_{i:S_i=0}L(Y_i)=O_p(n).$ Lemma \ref{lem-aux-1} yields that
	\begin{align*}
		\sup_v|\widehat{F}_{0}(v)-\widetilde{F}_{0}(v)|\leq  & \sup_v|
        \widetilde{F}_{0}(v+\tau)-\widetilde{F}_{0}(v-\tau)| \\
		 & + \frac{1}{n_0\tau}\sum_{i:S_i=0}\left|\widehat{f}(X_i,0)-f^*(X_i,0)\right|
         \\ \leq  & \sup_v|
        F^\ast_0(v+\tau)-F_{0}^\ast(v-\tau)| \\
		 & + \frac{1}{n_0\tau}\sum_{i:S_i=0}\left|\widehat{f}(X_i,0)-f^*(X_i,0)\right|
         \\ & + 2\sup_v|\widetilde{F}_{0}(v)-F_{0}^*(v)|
	\end{align*}
for any $\tau>0$. By Assumption \ref{assum: f} and the fact that $n_0/n\rightarrow^p r_0$,
\begin{align*}
\frac{1}{n_0\tau}\sum_{i:S_i=0}\left|\widehat{f}(X_i,0)-f^*(X_i,0)\right|=O_p(\tau^{-1}r_{n,d}).    
\end{align*}
On the other hand, by Assumption \ref{assum: F}, $\sup_v|
        F^\ast_0(v+\tau)-F_{0}^\ast(v-\tau)|=O(\tau).$
Using the fact that $\sup_v |\widetilde{F}_{0}(v)-F_{0}^*(v)|=O_p(n^{-1/2})$ and choosing $\tau=r_{n,d}^{1/2}$, we have $$\sup_v|\widehat{F}_{0}(v)-\widetilde{F}_{0}(v)|=O_p(r_{n,d}^{1/2}+n^{-1/2}).$$ 
Therefore, 
	\begin{align}\label{eq-uhat}
		\sup_{Q\in\cQ}\left|\sum_{i=1}^n\mathcal{L}(Q(\widehat{U}_i),Y_i)-\sum_{i=1}^n\mathcal{L}(Q(U_i),Y_i)\right|=O_p(nr_{n,d}^{1/2} + n^{1/2}).      
	\end{align}
	Next, by the definition of $\widehat{Q}$, we get
	\begin{align*}
		\sum_{i=1}^n\mathcal{L}(\widehat{Q}(\widehat{U}_i),Y_i) \leq \sum_{i=1}^n\mathcal{L}(\widetilde{Q}(\widehat{U}_i),Y_i).    
	\end{align*}
	Together with (\ref{eq-uhat}), we obtain
	\begin{align*}
		\sum_{i=1}^n\mathcal{L}(\widehat{Q}(U_i),Y_i) \leq \sum_{i=1}^n\mathcal{L}(\widetilde{Q}(U_i),Y_i) + O_p(nr_{n,d}^{1/2} + n^{1/2}),   
	\end{align*}
	i.e., $M_n(\widehat Q)\geq M_n(\widetilde{Q}) + O_p(r_{n,d}^{1/2} + n^{-1/2})$. Therefore, $I_2\rightarrow 0$, and the conclusion follows.
\end{proof}

\begin{proof}[Proof of Proposition \ref{prop-sep-con}]
Given $U$, we define
$h_U(\cdot)=h(\cdot;Q_0^*(U),Q_1^*(U))$. By the definition of $\widetilde{Q}$, we must have $h^{(1)}_U(\widetilde{Q}(U))=0$. As $h_U$ is strongly convex and $h^{(2)}_U(x)\geq 2\lambda$ for any $x$, we obtain 
\begin{align*}
h_U(Q(U)) - h_U(\widetilde{Q}(U)) \geq & h^{(1)}_{U}(\widetilde{Q}(U)) (Q(U)-\widetilde{Q}(U)) + \lambda (Q(U)-\widetilde{Q}(U))^2 
\\ =& \lambda (Q(U)-\widetilde{Q}(U))^2.    
\end{align*}
Taking expectations on both sides leads to the conclusion. It remains to verify $h^{(2)}_U(x)\geq 2\lambda$ when $\mathcal{L}$ is the squared or the cross-entropy loss. This is clearly true for the squared loss as $h^{(2)}_U(x)=2$ for all $x$. For the cross-entropy loss, note that
\begin{align*}
h^{(2)}_U(x)=\sum^{1}_{s=0}r_s\left\{\frac{Q^*_s(U)}{x^2}+\frac{1-Q_s^*(U)}{(1-x)^2}\right\}\geq \sum^{1}_{s=0}r_s\left\{Q_s^*(U) + 1 - Q_s^*(U)\right\}=1      
\end{align*}
for all $x\in (0,1)$.
\end{proof}

\begin{lemma}\label{lem-uhat}
Under Assumptions \ref{assum: f}-\ref{assum: F}, we have
\begin{align*}
\mE\left[|U-\widehat{U}|^2\big|S=s\right]=O(\{r_{n,d}\sqrt{\log(n)}\}^2+n^{-1}),\quad s=0,1.    
\end{align*}
Moreover, under the assumption in Remark \ref{rm1} and Assumption \ref{assum: F},
\begin{align*}
\mE\left[|U-\widehat{U}|^2\big|S=s\right]=O(r_{n,d}^{2m/(m+1)}+n^{-1}),\quad s=0,1.
\end{align*}
\end{lemma}
\begin{proof}[Proof of Lemma \ref{lem-uhat}]
Conditional on $S=0$ and using Assumption \ref{assum: F}, we have
\begin{align*}
|U-\widehat{U}|^2=& |\widehat{F}_{0}\circ \widehat{f}(X,0)-F_{0}^{\ast}\circ f^*(X,0)|^2 
\\ \leq & 3|\widehat{F}_{0}\circ \widehat{f}(X,0)-\widetilde{F}_{0}\circ \widehat{f}(X,0)|^2 + 3|\widetilde{F}_{0}\circ \widehat{f}(X,0)-F_{0}^{\ast}\circ \widehat{f}(X,0)|^2 
\\ & + 3| F_{0}^{\ast}\circ \widehat{f}(X,0)-F_{0}^{\ast}\circ f^*(X,0)|^2
\\ \leq & 3\sup_v|\widehat{F}_{0}(v)-\widetilde{F}_{0}(v)|^2 +3 \sup_v |\widetilde{F}_0(v)-F_0^\ast(v)|^2 + C |\widehat{f}(X,0)-f^\ast(X,0)|^2.
\end{align*}
We first derive the result under Assumption \ref{assum: f}. By Lemma \ref{lem-aux-1},
\begin{align*}
\sup_v|\widehat{F}_{0}(v)-\widetilde{F}_{0}(v)|^2\leq & C^2\tau^2
		 + \frac{C}{n_0^2}\left(\sum_{i:S_i=0}\mathbf{1}\left\{\left|\widehat{f}(X_i,0)-f^*(X_i,0)\right|>\tau\right\}\right)^2
        \\& + C\sup_v|\widetilde{F}_{0}(v)-F_{0}^*(v)|^2.
\end{align*}
Using the Dvoretzky–Kiefer–Wolfowitz inequality, it can be shown that 
$$\mE[\sup_v|\widetilde{F}_{0}(v)-F_{0}^*(v)|^2]=\int
\mP\left(\sup_t|\widetilde{F}_{0}(v)-F_{0}^*(v)|>\sqrt{u}\right)du=O(n^{-1}).$$ 
By Assumption \ref{assum: f}, 
\begin{align*}
\mE\left[\left(\sum_{i:S_i=0}\mathbf{1}\left\{\left|\widehat{f}(X_i,0)-f^*(X_i,0)\right|>\tau\right\}\right)^2\right]
\leq& n^2r_0 \mP\left(\left|\widehat{f}(X,0)-f^*(X,0)\right|>\tau|S=0\right)
\\ \leq& Cn^2r_0 \exp(-C\tau^2/r_{n,d}^2). 
\end{align*}
Setting $\tau=r_{n,d}\sqrt{\log(n)}$ leads to 
\begin{align*}
\mE\left[|U-\widehat{U}|^2\big|S=0\right]=O(\{r_{n,d}\sqrt{\log(n)}\}^2+n^{-1}). 
\end{align*}

Next, we derive the result under (\ref{as-moment}). Again by Lemma \ref{lem-aux-1},
\begin{align*}
\sup_v|\widehat{F}_{0}(v)-\widetilde{F}_{0}(v)|^2\leq & C^2\tau^2
		 + \frac{C}{n_0^2\tau^{2m}}\left(\sum_{i:S_i=0}\left|\widehat{f}(X_i,0)-f^*(X_i,0)\right|^m\right)^2
        \\& + C\sup_v|\widetilde{F}_{0}(v)-F_{0}^*(v)|^2.
\end{align*}
As $\mE[|\widehat{f}(X,0)-f^\ast(X,0)|^{2m}|S=0]=O(r_{n,d}^{2m})$, setting $\tau=r_{n,d}^{m/(m+1)}$ leads to 
\begin{align*}
\mE\left[|U-\widehat{U}|^2\big|S=0\right]=O(r_{n,d}^{2m/(m+1)}+n^{-1}).    
\end{align*}
\end{proof}

\begin{proof}[Proof of Theorem \ref{thm_rate}]
Write $\delta_n=\{r_{n,d}\sqrt{\log(n)}\}^{\frac{2k+1}{2k+3}}\vee n^{-\frac{2k+1}{2(2k+3)}}.$ We will prove the convergence order by showing that for any $\varepsilon>0$, there exists an $N>0$ such that
	\begin{align*}
	\mP(d(\widehat Q, \widetilde{Q}) > 2^{N} \delta_n) < \varepsilon.
	\end{align*}
	Define $\widehat M_n(Q) = -n^{-1}\sum_{i=1}^{n} \cL(Q(\widehat U_i), Y_i)$ and let
	$$S_{n,j} = \left\{Q \in \cQ: 2^j \delta_n< d(Q,\widetilde{Q}) \le 2^{j+1} \delta_n\right\}.$$
For an arbitrary constant $N>0$ and $\gamma>2^N\delta_n$,
	\begin{align*}
		\mP(d(\widehat Q, \widetilde{Q}) > 2^{N} \delta_n) \le& \mP(d(\widehat Q, \widetilde{Q}) > 2^{N} \delta_n, d(\widehat Q,\widetilde{Q}) \le \gamma) + \mP(d(\widehat Q,\widetilde{Q}) > \gamma)\\
		\le&  \sum_{j \ge N, 2^{j} \delta_n < \gamma} \mP(2^j \delta_n< d(\widehat Q, \widetilde{Q}) \le 2^{j+1} \delta_n) + o(1)\\
		=& \sum_{j \ge N, 2^{j} \delta_n < \gamma} \mP(\widehat Q \in S_{n,j}) + o(1).
	\end{align*}
	For any $Q \in S_{n,j}$, since $d(Q,\widetilde{Q}) > 2^j \delta_n$, Assumption \ref{assum:loss} implies that
	\begin{align*}
		M(\widetilde{Q}) - M(Q) \ge \lambda d^2 (Q,\widetilde{Q}) > \lambda 2^{2j} \delta_n^2.
	\end{align*}
	Since $\widehat Q$ maximizes $\widehat M_n(Q)$ over $\cQ$ , $\widehat Q \in S_{n,j}$ entails that
	\begin{align*}
		\sup_{Q \in S_{n,j}}\widehat M_n(Q) - \widehat M_n(\widetilde{Q}) \ge 0.
	\end{align*}
	Combining the above results, we obtain
	\begin{align}
		\mP(\widehat Q \in S_{n,j})
		\le& \mP\left(\sup_{Q \in S_{n,j}}\left\{ M(\widetilde{Q}) - M(Q) + \widehat M_n(Q)-\widehat M_n(\widetilde{Q}) \right\} > \lambda 2^{2j} \delta_n^2\right) \notag\\ 
		\le& \mP\left(\sup_{\left\{Q \in \cQ: d(Q,\widetilde{Q}) \le 2 ^{j+1} \delta_n \right\}} \{G_n(Q) + \Delta_n(Q)\} > \lambda 2^{2j} \delta_n^2\right)  \notag\\
		\le & \mP\left(\sup_{\left\{Q \in \cQ: d(Q,\widetilde{Q}) \le  2^{j+1} \delta_n\right\}} \Delta_n(Q) > \lambda 2^{2j-1} \delta_n^2\right) \label{peel: Delta_n}\\ 
		&+ \mP\left(\sup_{\left\{Q \in \cQ: d(Q,\widetilde{Q}) \le 2^{j+1} \delta_n\right\}} G_n(Q) > \lambda 2^{2j-1} \delta_n^2\right), \label{peel: G_n}
	\end{align}
where we have defined 
\begin{align*}
& \Delta_n(Q):=\widehat M_n(Q)- M_n(Q)+ M_n(\widetilde{Q}) - \widehat M_n(\widetilde{Q}),\\
& G_n(Q) := M_n(Q) - M\left( Q\right) - M_n(\widetilde{Q}) + M( \widetilde{Q}).
\end{align*}

We first control the term in (\ref{peel: G_n}). By Markov's inequality, 
\begin{align*}
\mP\left(\sup_{\left\{Q \in \cQ: d(Q,\widetilde{Q}) \le 2^{j+1} \delta_n\right\}} G_n(Q) > \lambda 2^{2j-1} \delta_n^2\right) 
\le \frac{1}{\lambda 2^{2j-1} \delta_n^2}\mE \left[\sup_{\left\{Q \in \cQ: d(Q,\widetilde{Q}) \le 2^{j+1} \delta_n\right\}}  \vert G_n(Q)\vert \right].
\end{align*}
Define the class of functions 
$$\mathcal{G} = \left\{ \cL \left\{Q(\cdot),\cdot\right\} - \cL \{\widetilde{Q}(\cdot),\cdot\}: Q \in \cQ, d(Q,\widetilde{Q}) \le 2^{j+1} \delta_n\right\}.$$ 
For $Q_1,Q_2\in \{Q \in \cQ: d(Q,\widetilde{Q}) \le 2^{j+1} \delta_n\}$, Assumption \ref{assum:loss} and the condition that $\sup_y|L(y)|\leq B$ imply that 
\begin{align*}
\left\| \cL(Q_1(\cdot),\cdot) - \cL(Q_2(\cdot),\cdot) \right\|_{\mP_{UY}}^2= & \mE\left[\{\cL(Q_1(U),Y) - \cL(Q_2(U),Y)\}^2\right]
\\ \leq & \mE\left[L(Y)^2\{Q_1(U)-Q_2(U)\}^2\right]
\\ \leq & B^2 d^2(Q_1,Q_2)\leq B^2 2^{2(j+1)}\delta_n^2.
\end{align*}
Following the argument in the proof of Theorem \ref{thm_consist}, 
let $Q_1,\dots,Q_{N_\varepsilon}$ be an $\varepsilon$-cover of $\{Q \in \cQ: d(Q,\widetilde{Q}) \le 2^{j+1} \delta_n\}$
under the metric $\|\cdot\|_2$, i.e., for any $Q\in \widetilde{\cQ}$ with $d(Q,\widetilde{Q}) \le 2^{j+1} \delta_n$, there exists a $Q_i$ such that $\|Q-Q_i\|_2 \leq \varepsilon$. Then the brackets $[\cL(Q_i(\cdot),\cdot)-B\varepsilon, \cL(Q_i(\cdot),\cdot)+B\varepsilon]$ for $1\leq i\leq N_\varepsilon$ form an $\varepsilon$-cover of $\mathcal{G}$ under the metric $\|\cdot\|_{\mP_{UY}}$. Thus the bracketing number of $\mathcal{G}$ satisfies that
\begin{align*}
N_{[\,]} \left(2B\varepsilon,\mathcal G, \|\cdot\|_{\mP_{UY}}\right) \le  N(\varepsilon,\{Q \in \cQ: d(Q,\widetilde{Q}) \le 2^{j+1} \delta_n\},\|\cdot\|_2)
\le A \left(\frac{2^{j+1} \delta_n}{\varepsilon}\right)^{1/k},
\end{align*}
where we used the fact that 
$$N(\varepsilon,\{Q \in \cQ: d(Q,\widetilde{Q}) \le \tau\},\|\cdot\|_2)\leq N(\varepsilon,\{Q \in \cQ: d(Q,\widetilde{Q}) \le \tau\},\|\cdot\|_\infty)
\le A \left(\frac{\tau}{\varepsilon}\right)^{1/k}.$$
As $\sup_y|L(y)|\leq B$, $\|\cL \left\{Q(\cdot),\cdot\right\} - \cL \{\widetilde{Q}(\cdot),\cdot\}\|_\infty\leq C$ for some constant $C>0.$ By Lemma 3.4.2 of \cite{van1996weak}, we have
\begin{align*}
\mE \left[\sup_{Q \in \cQ: d\left(Q,\widetilde{Q}\right) \le 2^{j+1}\delta_n} n^{1/2}\vert G_n(Q) \vert \right] \leq J_{[\,]}(B2^{j+1}\delta_n,\mathcal{G},\|\cdot\|_{\mP_{UY}})\left(1+\frac{J_{[\,]}(B2^{j+1}\delta_n,\mathcal{G},\|\cdot\|_{\mP_{UY}})}{B^22^{2(j+1)}\delta_n^2\sqrt{n}}C\right),
\end{align*}
where
\begin{align*}
J_{[\,]}(\xi,\mathcal{G},\|\cdot\|_{\mP_{UY}})=\int_{0}^{\xi}  \sqrt{1+\log N_{[\,]} \left(\varepsilon,\mathcal G, \|\cdot\|_{\mP_{UY}}\right)} d\varepsilon.    
\end{align*}
It implies that
\begin{align*}
\mE \left[\sup_{Q \in \cQ: d\left(Q,\widetilde{Q}\right) \le 2^{j+1}\delta_n} n^{1/2}\vert G_n(Q) \vert \right]\leq B2^{j+1}\delta_n\left(1+\frac{C}{B2^{j+1}\delta_n \sqrt{n}}\right).
\end{align*}
Therefore,
\begin{align*}
\mP\left(\sup_{\left\{Q \in \cQ: d(Q,\widetilde{Q}) \le 2^{j+1} \delta_n\right\}} G_n(Q) > \lambda 2^{2j-1} \delta_n^{2}\right) 
\le& \frac{1}{\lambda 2^{2j-1} \delta_n^2}  \mE \left[\sup_{Q \in \cQ: d\left(Q,\widetilde{Q}\right) 
\le 2^{j+1}\delta_n} \vert G_n(Q) \vert \right]
\\ \leq & \frac{C 2^{j+1}\delta_n}{n^{1/2}\lambda 2^{2j-1} \delta_n^2} 
\\ \leq & \frac{C}{ n^{1/2} 2^{j}\delta_n}.
\end{align*}

Next, we control the term in (\ref{peel: Delta_n}). Note that
\begin{align*}
\Delta_n(Q)
=& \frac{1}{n}\sum_{i=1}^{n}\left[ \left\{\cL(Q(U_i),Y_i) - \cL(Q(\widehat U_i),Y_i)\right\} -  \left\{\cL(\widetilde{Q}(U_i),Y_i) - \cL(\widetilde{Q}(\widehat U_i),Y_i)\right\}\right].
\end{align*}
By the smoothness property of $\cL(\cdot,\cdot)$ and for $Q\in \cQ$ with $d(Q,\widetilde{Q}) \le 2^{j+1} \delta_n$, the mean-value theorem implies that
\begin{align*}
&\left\{\cL(Q(U_i),Y_i) - \cL(Q(\widehat U_i),Y_i)\right\} -  \left\{\cL(\widetilde{Q}(U_i),Y_i) - \cL(\widetilde{Q}(\widehat U_i),Y_i)\right\}\\
=& \cL^{\prime} (\widetilde Q_{1,i}, Y_i) \left\{ Q(U_i)- Q(\widehat U_i) \right\} - \cL^{\prime} (\widetilde Q_{2,i}, Y_i) \left\{ \widetilde{Q}(U_i)-\widetilde{Q}(\widehat U_i) \right\} \\
= & \cL^{\prime} (\widetilde Q_{1,i}, Y_i) \left\{ Q(U_i)- Q(\widehat U_i) - \widetilde{Q}(U_i) + \widetilde{Q}(\widehat U_i) \right\} 
\\ & +\left\{\cL^{\prime} (\widetilde Q_{1,i}, Y_i) - \cL^{\prime} (\widetilde Q_{2,i}, Y_i)\right\}\left\{\widetilde{Q}(U_i) - \widetilde{Q}(\widehat U_i) \right\} 
\\ :=& I_{1,i}(Q) + I_{2,i}(Q)
\end{align*}
where $\widetilde Q_{1,i}$ lies between $Q(U_i)$ and $Q(\widehat U_i)$, $\widetilde Q_{2,i}$ lies between $\widetilde{Q}(U_i)$ and $\widetilde{Q}(\widehat U_i)$ and 
\begin{align*}
I_{1,i}(Q):=\cL^{\prime} (\widetilde Q_{1,i}, Y_i) \left\{ Q(U_i)- Q(\widehat U_i) - \widetilde{Q}(U_i) + \widetilde{Q}(\widehat U_i) \right\},\\
I_{2,i}(Q):=\left\{\cL^{\prime} (\widetilde Q_{1,i}, Y_i) - \cL^{\prime} (\widetilde Q_{2,i}, Y_i)\right\}\left\{\widetilde{Q}(U_i) - \widetilde{Q}(\widehat U_i) \right\}.
\end{align*}
By the mean-value theorem again, we get
\begin{align*}
|I_{1,i}(Q)|\leq \left|\cL^{\prime} (\widetilde Q_{1,i}, Y_i)\right| \left|Q^{(1)}(\breve{U}_i)-\widetilde{Q}^{(1)}(\breve{U}_i)\right| \left|U_i-\widehat{U}_i\right|,  
\end{align*}
where $\breve{U}_i$ is between $U_i$ and $\widehat{U}_i$. It thus implies that
\begin{align*}
\frac{1}{n}\sum^{n}_{i=1}|I_{1,i}(Q)| \leq & \frac{C}{n}\sum^{n}_{i=1} \left|Q^{(1)}(\breve{U}_i)-\widetilde{Q}^{(1)}(\breve{U}_i)\right| \left|U_i-\widehat{U}_i\right| \\
\leq & \frac{C}{n}\left(\sum^{n}_{i=1} \left|Q^{(1)}(\breve{U}_i)-\widetilde{Q}^{(1)}(\breve{U}_i)\right|^2\right)^{1/2} \left(\sum^{n}_{i=1}\left|U_i-\widehat{U}_i\right|^2\right)^{1/2}.
\end{align*}

To deal with the term $\sum^{n}_{i=1} \left|Q^{(1)}(\breve{U}_i)-\widetilde{Q}^{(1)}(\breve{U}_i)\right|^2$, we note that
\begin{align*}
\left|Q^{(1)}(\breve{U}_i)-\widetilde{Q}^{(1)}(\breve{U}_i)\right|^2\leq & 3|Q^{(1)}(\breve{U}_i)-Q^{(1)}(U_i)|^2+ 3|\widetilde{Q}^{(1)}(\breve{U}_i)-\widetilde{Q}^{(1)}(U_i)|^2 
\\&+ 3 \left|Q^{(1)}(U_i)-\widetilde{Q}^{(1)}(U_i)\right|^2
\\ \leq & C|\widehat{U}_i-U_i|^2 + 3|Q^{(1)}(U_i)-\widetilde{Q}^{(1)}(U_i)|^2,
\end{align*}
where we use the assumption that $\sup_{Q\in \widetilde{\cQ}}\|Q^{(2)}\|_\infty\leq C$. Therefore, 
\begin{align*}
\frac{1}{n}\sum^{n}_{i=1}|I_{1,i}(Q)| \leq \frac{C}{n}\sum^{n}_{i=1}\left|U_i-\widehat{U}_i\right|^2
+ \frac{C}{n}\left(\sum^{n}_{i=1} \left|Q^{(1)}(U_i)-\widetilde{Q}^{(1)}(U_i)\right|^2\right)^{1/2} \left(\sum^{n}_{i=1}\left|U_i-\widehat{U}_i\right|^2\right)^{1/2}.
\end{align*}
By Lemma \ref{lem-uhat}, we have
\begin{align*}
\sup_i\sup_{s=0,1}\mE\left[\left|U_i-\widehat{U}_i\right|^2\Big|S_i=s\right]=O(\{r_{n,d}\sqrt{\log(n)}\}^2 + n^{-1}),  
\end{align*}
which implies that
\begin{align*}
&\mP\left(\sup_{\left\{Q \in \cQ: d(Q,\widetilde{Q}) \le 2^{j+1} \delta_n\right\}}\frac{1}{n}\sum^{n}_{i=1}|I_{1,i}(Q)|>\lambda 2^{2j-2}\delta_n^2\right) 
 \\ \leq & \frac{\mE[\sup_{\left\{Q \in \cQ: d(Q,\widetilde{Q}) \le 2^{j+1} \delta_n\right\}}n^{-1}\sum^{n}_{i=1}|I_{1,i}(Q)|]}{\lambda 2^{2j-2}\delta_n^2}
\\ \leq& C\frac{\{r_{n,d}\sqrt{\log(n)}\}^2+n^{-1} + D_{1,j}^{1/2}\sqrt{\{r_{n,d}\sqrt{\log(n)}\}^2+n^{-1}}}{\lambda 2^{2j-2}\delta_n^2}.
\end{align*}
Here, we define  
\begin{align*}
D_{1,j} = \mE\left[\sup_{\left\{Q \in \cQ: d(Q,\widetilde{Q}) \le 2^{j+1} \delta_n\right\}}n^{-1}\sum^{n}_{i=1}|Q^{(1)}(U_i)-\widetilde{Q}^{(1)}(U_i)|^2\right].    
\end{align*}
The Gagliardo–Nirenberg interpolation inequality yields that
%\begin{align*}
%\mE[|Q^{(1)}(U)-\widetilde{Q}^{(1)}(U)|^2]=d^2(Q^{(1)},\widetilde{Q}^{(1)})\leq Cd^{2(1-1/k)}(Q,\widetilde{Q}).    
%\end{align*}
\begin{align*}
\mE[|Q^{(1)}(U)-\widetilde{Q}^{(1)}(U)|^2]=d^2(Q^{(1)},\widetilde{Q}^{(1)})\leq Cd^{2(2k-1)/(2k+1)}(Q,\widetilde{Q})   
\end{align*}
due to that $\sup_{Q\in\mathcal{Q}}\|Q^{(k)}\|_{\infty} <\infty$.
Consider the class of functions $\widetilde{\mathcal{G}}=\{|Q^{(1)}(\cdot)-\widetilde{Q}^{(1)}(\cdot)|^2: Q \in \cQ, d(Q^{(1)},\widetilde{Q}^{(1)}) \le 2^{(j+1)(2k-1)/(2k+1)} \delta_n^{(2k-1)/(2k+1)}\}$. For any $g\in \widetilde{\mathcal{G}}$, $\|g\|_\infty \leq C$ and 
\begin{align*}
N_{[\,]} \left(\varepsilon,\widetilde{\mathcal G}, \|\cdot\|_{2}\right)
\le C\left(\frac{2^{(j+1)(2k-1)/(2k+1)} \delta_n^{(2k-1)/(2k+1)}}{\varepsilon}\right)^{1/(k-1)}.  
\end{align*}
By Lemma 3.4.2 of \cite{van1996weak}, we get
\begin{align*}
D_{1,j} \leq & \mE\left[\sup_{\left\{Q \in \cQ: d(Q,\widetilde{Q}) \le 2^{j+1} \delta_n\right\}}n^{-1}\sum^{n}_{i=1}|Q^{(1)}(U_i)-\widetilde{Q}^{(1)}(U_i)|^2-\mE[|Q^{(1)}(U_i)-\widetilde{Q}^{(1)}(U_i)|^2]\right] 
\\& +\sup_{\left\{Q \in \cQ: d(Q,\widetilde{Q}) \le 2^{j+1} \delta_n\right\}}\mE[|Q^{(1)}(U_i)-\widetilde{Q}^{(1)}(U_i)|^2]
\\ \leq & C\frac{2^{(j+1)(2k-1)/(2k+1)}\delta_n^{(2k-1)/(2k+1)}}{\sqrt{n}}+C2^{(j+1)2(2k-1)/(2k+1)}\delta_n^{2(2k-1)/(2k+1)}
\\ \leq& C 2^{(j+1)2(2k-1)/(2k+1)}\delta_n^{2(2k-1)/(2k+1)}.
\end{align*}
Because $\sup_{Q\in\cQ}\|Q^{(1)}\|_\infty\leq C$ and $\sup_{|x|\leq R_2}|\mathcal{L}''(x,Y)|\leq C$ almost surely, we have
\begin{align*}
|I_{2,i}(Q)|\leq &  C\left|U_i-\widehat{U}_i\right|\left|\widetilde Q_{1,i}-\widetilde Q_{2,i}\right| 
\\ \leq & C\left|U_i-\widehat{U}_i\right| \left\{\left|Q(U_i)-Q(\widehat{U}_i)\right|+\left|Q(U_i)-\widetilde{Q}(U_i)\right|+\left|\widetilde{Q}(U_i)-\widetilde{Q}(\widehat{U}_i)\right|\right\}
\\ \leq & C\left\{\left|U_i-\widehat{U}_i\right|^2+\left|U_i-\widehat{U}_i\right|\left|Q(U_i)-\widetilde{Q}(U_i)\right|\right\}.
\end{align*}
Hence, 
\begin{align*}
\frac{1}{n}\sum^{n}_{i=1}|I_{2,i}(Q)| \leq & \frac{C}{n}\sum^{n}_{i=1} \left|U_i-\widehat{U}_i\right|^2 + \frac{C}{n}\left(\sum^{n}_{i=1} \left|Q(U_i)-\widetilde{Q}(U_i)\right|^2\right)^{1/2} \left(\sum^{n}_{i=1}\left|U_i-\widehat{U}_i\right|^2\right)^{1/2}. 
\end{align*}
Again, by Lemma \ref{lem-uhat}, we have
\begin{align*}
&\mP\left(\sup_{\left\{Q \in \cQ: d(Q,\widetilde{Q}) \le 2^{j+1} \delta_n\right\}}\frac{1}{n}\sum^{n}_{i=1}|I_{2,i}(Q)|>\lambda 2^{2j-2}\delta_n^2\right) 
 \\ \leq & \frac{\mE[\sup_{\left\{Q \in \cQ: d(Q,\widetilde{Q}) \le 2^{j+1} \delta_n\right\}}n^{-1}\sum^{n}_{i=1}|I_{2,i}(Q)|]}{\lambda 2^{2j-2}\delta_n^2}
\\ \leq& C\frac{\{r_{n,d}\sqrt{\log(n)}\}^2+n^{-1} + D_{2,j}^{1/2}\sqrt{\{r_{n,d}\sqrt{\log(n)}\}^2+n^{-1}}}{\lambda 2^{2j-2}\delta_n^2}.
\end{align*}
Here, we define  
\begin{align*}
D_{2,j} = \mE\left[\sup_{\left\{Q \in \cQ: d(Q,\widetilde{Q}) \le 2^{j+1} \delta_n\right\}}n^{-1}\sum^{n}_{i=1}|Q(U_i)-\widetilde{Q}(U_i)|^2\right].    
\end{align*}
Consider the class of functions $\breve{\mathcal{G}}=\{|Q(\cdot)-\widetilde{Q}(\cdot)|^2: Q \in \cQ, d(Q,\widetilde{Q}) \le 2^{j+1} \delta_n\}$. For any $g\in \breve{\mathcal{G}}$, $\|g\|_\infty \leq C$ and 
\begin{align*}
N_{[\,]} \left(\varepsilon,\breve{\mathcal G}, \|\cdot\|_{2}\right)
\le C\left(\frac{2^{j+1} \delta_n}{\varepsilon}\right)^{1/k}.  
\end{align*}
Again, by Lemma 3.4.2 of \cite{van1996weak}, we get
\begin{align*}
D_{2,j} \leq & \mE\left[\sup_{\left\{Q \in \cQ: d(Q,\widetilde{Q}) \le 2^{j+1} \delta_n\right\}}n^{-1}\sum^{n}_{i=1}|Q(U_i)-\widetilde{Q}(U_i)|^2-\mE[|Q(U_i)-\widetilde{Q}(U_i)|^2]\right] +2^{2(j+1)}\delta_n^2
\\ \leq & C\frac{2^{j+1}\delta_n}{\sqrt{n}}+2^{2(j+1)}\delta_n^2\leq C 2^{2(j+1)}\delta_n^2.
\end{align*}
Combining the above results, for any $\varepsilon>0$, by taking $N$ large enough, we have
\begin{align*}
&P(d(\widehat Q, \widetilde{Q}) > 2^N \delta_n) 
\\ \le & \sum_{j \ge N, 2^{j} \delta_n < \gamma}C\frac{\{r_{n,d}\sqrt{\log(n)}\}^2+n^{-1} + 2^{(j+1)(2k-1)/(2k+1)}\delta_n^{(2k-1)/(2k+1)}\sqrt{\{r_{n,d}\sqrt{\log(n)}\}^2+n^{-1}}}{\lambda 2^{2j-2}\delta_n^2}
\\&+\frac{C}{n^{1/2}\delta_n} \sum_{j \ge N, 2^{j} \delta_n < \gamma} 2^{-j} 
\\ \leq & \varepsilon.
\end{align*}
\end{proof}

\begin{proof}[Proof of Theorem \ref{thm_ffair}]
Recall the definition of the $L_2$ estimation error:
\begin{align*}
    d(\widehat{f}_{\text{fair}}, f^*_{\text{fair}}) &= \sqrt{\mathbb{E}_{X, S}\left[\left(\widehat{f}_{\text{fair}}(X, S) - f^*_{\text{fair}}(X , S) \right)^2\right]} \\
    &= \sqrt{\sum_{s \in \{0,1\}} r_s \mathbb{E}_{X^s}\left[\left(\widehat{f}_{\text{fair}}(X^s, s) - f^*_{\text{fair}}(X^s , s) \right)^2\right]}.
\end{align*}
Using the inequality $\sqrt{a+b} \le \sqrt{a} + \sqrt{b}$, we bound this by:
\begin{align*}
    d(\widehat{f}_{\text{fair}}, f^*_{\text{fair}}) &\leq \sum_{s \in \{0,1\}} \sqrt{r_s} \sqrt{\mathbb{E}_{X^s}\left[\left(\widehat{f}_{\text{fair}}(X^s, s) - f^*_{\text{fair}}(X^s , s) \right)^2\right]} \\
    &=: \sum_{s \in \{0,1\}} \sqrt{r_s} \|\widehat{f}_{\text{fair}} - f^*_{\text{fair}} \|_{P_{X^s}}.
\end{align*}
We decompose the error term inside the norm as follows:
\begin{align*}
    \widehat{f}_{\text{fair}}(X^s, s) - f^*_{\text{fair}}(X^s , s)
    &= \widehat{Q}\left(\widehat{F}_s(\widehat{f}(X^s, s))\right) - \widetilde{Q}\left(F_s^*(f^*(X^s, s))\right) \\
    &= \left[ \widehat{Q}\left(\widehat{F}_s(\widehat{f}(X^s, s))\right) - \widehat{Q}\left(F_s^*(f^*(X^s, s))\right) \right] \\
    &\quad + \left[ \widehat{Q}\left(F_s^*(f^*(X^s, s))\right) - \widetilde{Q}\left(F_s^*(f^*(X^s, s))\right) \right].
\end{align*}
Applying the triangle inequality to the norm $\|\cdot\|_{P_{X^s}}$:
\begin{align*}
    \|\widehat{f}_{\text{fair}} - f^*_{\text{fair}} \|_{P_{X^s}}
    &\leq \left\| \widehat{Q} \circ \widehat{F}_s \circ \widehat{f} - \widehat{Q} \circ F_s^* \circ f^* \right\|_{P_{X^s}} + \left\| \widehat{Q} \circ F_s^* \circ f^* - \widetilde{Q} \circ F_s^* \circ f^* \right\|_{P_{X^s}}.
\end{align*}
We analyze these two terms separately.
For the second term, recall that $U = F_s^*(f^*(X^s, s)) \sim \text{Unif}[0,1]$. Thus, this term corresponds exactly to the $L_2$ distance between the quantile functions:
\begin{align*}
    \left\| \widehat{Q} \circ F_s^* \circ f^* - \widetilde{Q} \circ F_s^* \circ f^* \right\|_{P_{X^s}} = \|\widehat{Q}(U) - \widetilde{Q}(U)\|_2 = d(\widehat{Q}, \widetilde{Q}).
\end{align*}
For the first term, Assumption \ref{assum: Q} implies that any function in $\mathcal{Q}$ (including the estimator $\widehat{Q}$) has a first derivative uniformly bounded by $R_1$. Consequently, $\widehat{Q}$ is Lipschitz continuous with constant $R_1$:
\begin{align*}
    \left\| \widehat{Q} \circ \widehat{F}_s \circ \widehat{f} - \widehat{Q} \circ F_s^* \circ f^* \right\|_{P_{X^s}}
    &\leq R_1 \left\| \widehat{F}_s \circ \widehat{f} - F_s^* \circ f^* \right\|_{P_{X^s}} \\
    &= R_1 \sqrt{\mathbb{E}[|\widehat{U} - U|^2 | S=s]}.
\end{align*}
Combining these bounds yields:
\begin{align*}
    d(\widehat{f}_{\text{fair}}, f^*_{\text{fair}}) \leq \sum_{s \in \{0,1\}} \sqrt{r_s} \left( R_1 \sqrt{\mathbb{E}[|\widehat{U} - U|^2 | S=s]} + d(\widehat{Q}, \widetilde{Q}) \right).
\end{align*}
From Lemma \ref{lem-uhat}, the first term is of order $O_p(r_{n,d}\sqrt{\log n} + n^{-1/2})$.
From Theorem \ref{thm_rate}, the second term $d(\widehat{Q}, \widetilde{Q})$ converges at the rate $O_p\left(\{r_{n,d}\sqrt{\log n}\}^{\frac{2k+1}{2k+3}} \vee n^{-\frac{2k+1}{2(2k+3)}}\right)$.
Since $\frac{2k+1}{2k+3} < 1$, the convergence rate of $d(\widehat{Q}, \widetilde{Q})$ is slower and thus dominates the error bound as $r_{n,d} \to 0$.
Therefore,
$$d(\widehat{f}_{\text{fair}}, f^*_{\text{fair}}) = O_p\left(\{r_{n,d}\sqrt{\log(n)}\}^{\frac{2k+1}{2k+3}}\vee n^{-\frac{2k+1}{2(2k+3)}}\right).$$
The result under the conditions of Remark \ref{rm1} follows analogously by substituting the corresponding moment-based rates.
\end{proof}

% Add this to the Appendix
\subsection{Extension to multiple protected groups} \label{app:multi_group}

In the general setting where the sensitive attribute $S$ takes values in $\mathcal{S}=\{0,1,\dots,K-1\}$ with probabilities $r_s = P(S=s)$, the framework extends naturally. The Fair Regression Problem seeks a predictor $f$ such that $f(X^s, s) \overset{d}{=} f(X^{s'}, s')$ for all $s, s'$.

The optimal fair predictor is given by $f(x,s) = \widetilde{Q} \circ F_s^* \circ f^*(x,s)$, where $\widetilde{Q}$ minimizes the aggregate risk across all groups. Conditional on $U=u$, $\widetilde{Q}(u)$ is obtained by solving the following pointwise optimization problem:
$$
\widetilde{Q}(u) = \argmin_{q \in \mathbb{R}} \sum_{s=1}^{K} r_s \mathbb{E}[\mathcal{C}(Q_s^*(U), q) \mid U=u].
$$
This formulation allows us to determine the explicit form of the optimal fair predictor for various regression tasks:

\begin{itemize}
    \item \textbf{Squared Loss ($L_2$):} Under the additive error model $Y = f^*(X,S) + \varepsilon$, the cost function corresponds to the squared loss. The minimizer of the weighted sum of squared differences is simply the weighted average. Thus, $\widetilde{Q}(u)$ is the probability-weighted average of the group-specific latent quantiles:
    $$ \widetilde{Q}(u) = \sum_{s=1}^{K} r_s Q_s^*(u). $$

    \item \textbf{Binary Classification (Cross-Entropy):} For binary classification where $Y|X,S \sim \text{Bernoulli}(f^*(X,S))$, the minimizer of the aggregate cross-entropy loss is the expected probability. Thus, $\widetilde{Q}(u)$ is the weighted average of the latent probabilities:
    $$ \widetilde{Q}(u) = \sum_{s=1}^{K} r_s Q_s^*(u). $$

    \item \textbf{Quantile Regression (Pinball Loss):} Consider the pinball loss $\mathcal{L}(q, y) = \rho_\tau(y-q)$ for a target quantile level $\tau \in (0,1)$. The optimal value $\widetilde{Q}(u)$ minimizes the weighted sum of expected pinball losses. This implies that $\widetilde{Q}(u)$ is the $\tau$-th quantile of the mixture distribution $\sum_{s=1}^K r_s F_s(\cdot|Q_s^*(u))$. It is defined implicitly by the equation:
    $$ \sum_{s=1}^{K} r_s F_s(\widetilde{Q}(u) \mid Q_s^*(u)) = \tau. $$

    \item \textbf{Robust Regression (Huber Loss):} For the Huber loss with parameter $\delta$, $\widetilde{Q}(u)$ corresponds to the robust M-estimator of location for the mixture distribution. It is the solution to the stationarity equation:
    $$ \sum_{s=1}^{K} r_s \mathbb{E}\left[ \psi_\delta(Y - \widetilde{Q}(u)) \Big| S=s, f^*(X,s)=Q_s^*(u) \right] = 0, $$
    where $\psi_\delta$ is the Huber influence function.
\end{itemize}

\end{document}